\documentclass[12pt]{article}

\usepackage[utf8]{inputenc}
\usepackage[T1]{fontenc}
\usepackage{authblk} 
\usepackage{amsmath,amsfonts,amssymb,amsthm}
\usepackage{color}
\usepackage[dvipsnames]{xcolor} 
\colorlet{mygreen}{green!60!black} 
\usepackage{graphicx} 
\usepackage{hyperref}

\colorlet{myblue}{blue!75!black}
\hypersetup{
    urlcolor=myblue}

\newcommand{\commeurl}[1]{\emph{\textcolor{myblue}{#1}}}

\usepackage{fancyhdr}
\usepackage{titling}
\usepackage{tocloft}  

\usepackage{tikz}
\usetikzlibrary{shapes} 
\usetikzlibrary{shapes.symbols}
\usetikzlibrary{decorations.pathreplacing} 
\usetikzlibrary{positioning} 

\usepackage{float}  
\usepackage{booktabs} 
\usepackage{setspace} 
\usepackage{scrextend} 

\usepackage{enumitem} 

\usepackage{geometry} 
\usepackage[font=small,labelfont=bf]{caption} 
\theoremstyle{definition}
\newtheorem{definition}{Definition}
\newtheorem{example}{Example}
\newtheorem{corollary}{Corollary}
\newtheorem{remark}{Remark}

\theoremstyle{theorem}
\newtheorem{theorem}{Theorem}
\newtheorem{proposition}{Proposition}

\def \be {\begin{equation}} 
\def \ee {\end{equation}}
\def \bes {\begin{equation*}}
\def \ees {\end{equation*}}
\def \baa {\begin{align}}
\def \eaa {\end{align}}
\def \baas {\begin{align*}}
\def \eaas {\end{align*}}
\def \bea {\begin{eqnarray}}
\def \eea {\end{eqnarray}}
\def \beas {\begin{eqnarray*}}
\def \eeas {\end{eqnarray*}}

\newcommand{\nib}[1]{\noindent\textbf{#1}}
\newcommand{\nit}[1]{\noindent\textit{#1}}

\newcommand{\ct}[1]{\begin{center}\textit{#1}\end{center}}
\newcommand{\st}{\,\colon\,} 

\newcommand{\deq}{\stackrel{\text{\tiny{def}}}{=}}
\newcommand{\re}[1]{\mathfrak{Re}\left(#1\right)}

\newcommand{\given}{\rvert} 
\newcommand{\pbfrac}[2]{\mbox{$\mbox{}^{#1}\!/_{#2}$}}
\newcommand{\finite}[2]{{#1}_{[#2]}} 

\newcommand{\halts}{\searrow} 
\newcommand{\nohalt}{\nearrow} 
\newcommand{\len}[1]{|#1|} 

\newcommand{\fat}{FAT} 
\newcommand{\fats}{FATs} 

\newcommand{\lecture}[2]{\noindent \href{#2}{\commeurl{[Lecture #1]}}} 
\newcommand{\pscom}[1]{\noindent \textit{#1}} %

\title{Lecture Notes on \\
{Algorithmic Information Theory}\vspace{0pt} }
\author{\vspace{0pt} Charles Alexandre Bédard}
\affil{École de technologie supérieure \vspace{-0pt}\\
\small \emph{charles.alexandre.bedard@etsmtl.ca}}
\date{\vspace{0pt} April 2025}

\begin{document}
\begin{titlepage}
    \maketitle
    \thispagestyle{empty} 
\end{titlepage}

\tableofcontents
\newpage


\section*{Preface}
\sloppy
These lecture notes on Algorithmic Information Theory were constructed in conjunction with the graduate course I taught at Universit\`a della Svizzera italiana in the spring of~2023. 

The course is intended for graduate students and researchers seeking a self-contained journey from the foundations of computability theory to prefix complexity and the information-theoretic limits of formal systems.
My exposition ignores boundaries between computer science, mathematics, physics, and philosophy, which I consider essential when explaining inherently multidisciplinary fields.

Recordings of the lectures are available online---click on the following link: \mbox{\href{https://www.youtube.com/playlist?list=PLdL1rYhN7DtjGYIfKpQBUuzlUGGet4Fm8}{\commeurl{[Full lecture series]}}}. 
Links to individual lectures are also provided throughout the notes.

The course is based on many books, most notably the detailed textbook by Li and Vit\'anyi~\cite{Li2019}, which readers should also consult for extensive references to the primary scientific literature.
I enjoyed the explanatory approach of Cover and Thomas~\cite{cover2006elements}, and I gained insights from Hutter~\cite{hutter2005universal} and Claude~\cite{calude2013information}.
Chaitin's \textit{Meta Math!}~\cite{chaitin2004meta}, both profound and accessible, also inspired the content of the course. 


As his Ph.D. student, I was introduced by Gilles Brassard to the marvels of theoretical computer science through many open discussions, for which I am deeply grateful. 
My interest in algorithmic information theory was then sparked by reading Bennett~\cite{bennett1979random} and Chaitin~\cite{chaitin2007}.

I am especially grateful to Stefan Wolf, who encouraged me to design and teach this course, and with whom I had many valuable conversations---both on this topic and beyond.
Finally, I am immensely grateful to the students who attended the lectures. 
Their enthusiasm and curiosity were potent motivators and led to countless stimulating discussions.

%
%

%
%
%
%
%
%
%

%

\newpage

\section{Introduction}

\lecture{AIT-1}{https://www.youtube.com/watch?v=gTic1BLEjaw}
\vspace{5pt}

Algorithmic information theory (AIT) is the meeting between Turing's theory of computation and Shannon's theory of information.

AIT 
roots the concept of information in computation rather than probability. 
Its central mathematical construct, \emph{algorithmic complexity}, is a universal quantification of information content in individual digital objects.

\begin{itemize}[itemsep=1pt, topsep=2pt]
\item[--] \emph{Quantification}, not qualification.
\item[--] \emph{Digital objects} are those describable by strings of symbols. 
\item[--] \emph{Individual}, i.e. we take the object as is; we do not posit a set or a distribution to which the object is assumed to belong. 
\item[--] \emph{Universal}, because any universal computer gives roughly the same measure of algorithmic complexity.
\end{itemize}

Algorithmic complexity was introduced independently by Solomonoff (1964) 
Kolmogorov (1965) and
Chaitin (1966).
Their work was broadly motivated---respectively---by induction (Solomonoff), information (Kolmogorov) and randomness (Chaitin).
AIT is still being developed.
See Fig.~\ref{fig:ITAIT} for a contrast with Shannon's information theory.

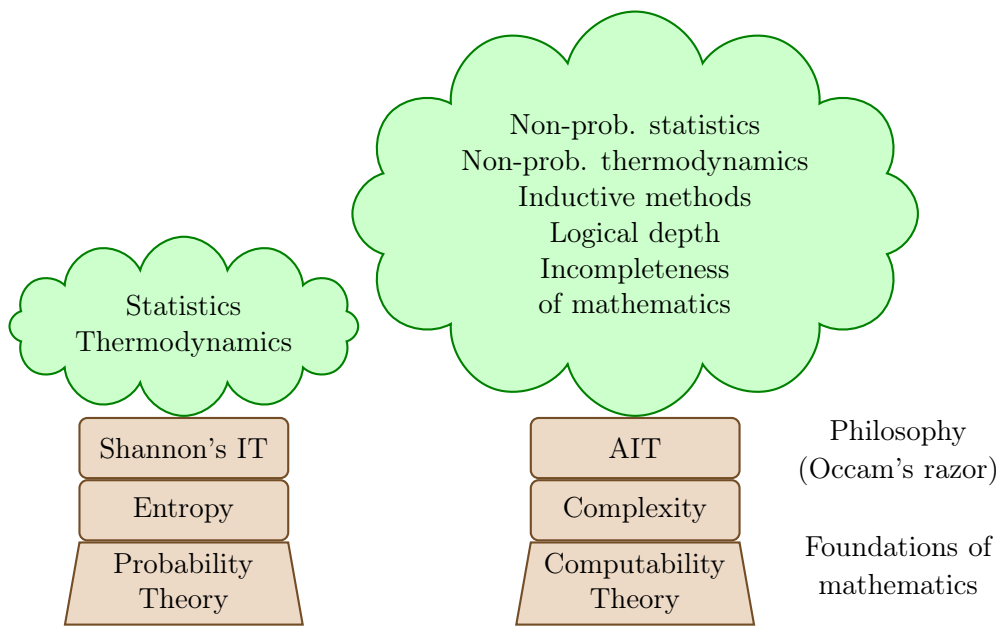
\begin{figure}[H]
\begin{tikzpicture}[
    font=\small,
    thick,
    root/.style={
      trapezium,
      trapezium left angle=100,
      trapezium right angle=100,
      shape border rotate=180,  
      draw=brown!60!black,
      fill=brown!30,
      text width=2.5cm,
      align=center,
      minimum height=1cm  
    },
    trunk/.style={
      rectangle,
      rounded corners=3pt,
      draw=brown!60!black,
      fill=brown!30,
      text width=2.5cm,  
      align=center,
      minimum height=0.8cm
    },
    midbox/.style={
      rectangle,
      rounded corners=3pt,
      draw=brown!60!black,
      fill=brown!30,
      text width=2.5cm,
      align=center,
      minimum height=0.8cm
    },
    leafcloud/.style={
      cloud,
      cloud puffs=12,
      cloud ignores aspect,
      minimum width=3cm,
      draw=green!50!black,
      fill=green!20,
      align=center
    }
]

\node[root] (roots-shannon) at (0,0)
  {Probability\\Theory};

\node[trunk, anchor=south] (trunk-shannon) at (roots-shannon.north)
  {Entropy};

\node[midbox, anchor=south] (top-shannon) at (trunk-shannon.north)
  {Shannon's IT};

\node[leafcloud, anchor=south] (cloud-shannon) at (top-shannon.north)
  {Statistics\\Thermodynamics};

\node[root] (roots-ait) at (6,0)
  {Computability\\Theory};

\node[trunk, anchor=south] (trunk-ait) at (roots-ait.north)
  {Complexity};

\node[midbox, anchor=south] (top-ait) at (trunk-ait.north)
  {AIT};

\node[leafcloud, anchor=south] (cloud-ait) at (top-ait.north)
  {Non-prob. statistics\\%
   Non-prob. thermodynamics\\Inductive methods\\%
   Logical depth\\Incompleteness\\%
   of mathematics};
   
\node (background1) at (9.5,2) {Philosophy};
\node (background2) at (9.5,1.5) {(Occam's razor)};

\node (background3) at (9.5,.5) {Foundations of};
\node (background4) at (9.5,0) {mathematics};

\end{tikzpicture}
\caption{Contrasting Shannon's information theory with AIT.}
\label{fig:ITAIT}
\end{figure}

\subsection{Hilbert's Theory of Everything}

Hilbert expressed the belief that if mathematics is to provide absolute and complete certainty, then there ought to be a formal axiomatic theory for all of mathematics. 
There should be a finite set of axioms from which all mathematical truths can be proven by mechanically applying the rules of logic.
In practice, today, the closest we have to such a theory is the Zermelo-Fraenkel set theory with the axiom of choice (ZFC), together with first-order logic.

Hilbert did not suggest that mathematics should only be done formally, but that it could, so matters of opinion and ambiguity can be eliminated from mathematics.


In 1931, Gödel puts an end to Hilbert's program.
He shows that in all logically consistent formal axiomatic systems that are strong enough to encapsulate the arithmetic of whole numbers, there are true and unprovable statements.
%
%
%
%
%
%
%
%
%
%
%
%
%

%

\subsection{Berry's Paradox}\label{sec:bp}

All natural numbers can be described unambiguously by (at least) a statement in English. 
For instance, $52$ is described by ``fifty-two'', $10 000 000 000$ is described by ``ten to the ten'' and $637 291 037$ is described by ``the number's digits are six, three, seven, two, nine, one, zero, three, seven.'' 
The description ``fifty-two'' was coined using 3 syllables, ``ten to the ten'', 4 syllables, and the last one, 18 syllables.
Consider the set $S$ of all numbers that can be described by 20 syllables or less.
There is a finite number~$s$ of syllables in English, and therefore there are at most $s^{21}$ elements in $S$. 
Therefore, there are infinitely many numbers in $ \mathbb N \backslash S$, and in particular, it has a smallest element; let us denote it $k_{20}$.
In other words,~$k_{20}$ is
\ct{``The smallest number that cannot be described by twenty syllables or less.''}
It has just been described by 19 syllables, so $k_{20} \not\in S$ and $k_{20} \in S$. A contradiction.

The formalization of algorithmic complexity not only solves the paradox but it turns it into proof of the incompleteness of mathematics---Chaitin's.

\newpage
\section{Strings and Codes}
\lecture{AIT-2}{https://youtu.be/0WdY376QJQU}


\subsection{Strings}
A set~$\mathcal A =  \{a_i\}_{i\in \mathcal I}$ of symbols is an \emph{alphabet}.
It can be finite, in which case~$\mathcal I= \{1,2, \ldots, N \}$ or infinite,~$\mathcal I=  \mathbb N$.

Symbols can be \emph{concatenated} into \emph{strings} (also called \emph{words}).
The set of all finite strings over an alphabet~$\mathcal A$ is denoted~$\mathcal A^*$.

The \emph{binary alphabet} is $\{0,1\}$ and
$$\{0,1\}^*= \{ \epsilon, 0, 1, 00, 01, 10, 11, 000, \ldots \}\,,$$
where $\epsilon$ is the \emph{empty string}.
The above set is displayed in the \emph{lexicographical ordering}.

The following pairing defines a bijection between $\{0,1\}^*$ and $\mathbb N$:
\bea \label{eq:bijection} 
\{0,1\}^* &\leftrightarrow& \mathbb N\nonumber \\
\epsilon &\leftrightarrow& 1\nonumber \\
0 &\leftrightarrow& 2 \nonumber \\
1 &\leftrightarrow& 3 \nonumber \\
00 &\leftrightarrow& 4 \\
01 &\leftrightarrow& 5 \nonumber \\
10 &\leftrightarrow& 6 \nonumber \\
11 &\leftrightarrow& 7 \nonumber \\
000 &\leftrightarrow& 8 \nonumber \,.
\eea
The relationship can be computed by observing that $1x$ is the binary representation of the number associated with the string $x$.
This relationship is often implied, so a string might be thought to be a number and vice versa; the context will determine. 

The \emph{length} of a string~$x$ is denoted $\len x$.
The \emph{ceiling} and \emph{floor} of a real number~$\alpha$ are respectively denoted~$\lceil \alpha \rceil$ and~$\lfloor \alpha \rfloor$.

\begin{proposition}\label{prop:length}
$\forall x \in \{0,1\}^*$, $\len x = \lfloor \log x \rfloor$.
\end{proposition}
\begin{proof}
Observe first that on the left-hand side,~$x$ is a string because we can't take the length of a number, whereas on the right-hand side, it is a number because we can't take the log of a string.
Since the binary representation of the number~$x$ is~$1x$, we have that 
$
x= 1 \cdot 2^{\len x} + r
$ with $r < 2^{\len x}$. Thus 
$
\lfloor \log x \rfloor =  \len x
$.
\end{proof}

A string $y$ is a \emph{prefix} of a string $x$ if $x=yz$ for some $z \in \mathcal A^*$. 
A set $S\subset \mathcal  A^*$ is \emph{prefix-free} if for all $x, y \in S$ such that $x \neq y$, $x$ is not a prefix of $y$.

\begin{example}
Consider the alphabet of digits $\mathcal D = \{0, 1,2,3,4,5,6,7,8,9 \}$.
The set of strings
$\{1,2,3,4,5,6,7,8,9,10,\ldots \}$ is not prefix-free. Yet with the alphabet $\mathcal D \cup \{*\}$, the set
$\{1*,2*,3*,4*,5*,6*,7*,8*,9*,10*,\ldots \}$ is prefix-free.
Similarly, $\{10,110,1110,11110, \ldots\}$ is a prefix-free set of bit strings.
\end{example}

\subsection{Codes}
\begin{definition}
A \emph{code} is a mapping 
$$
E: \mathcal A \to \mathcal B^*
$$
where~$\mathcal A$ is the \emph{source alphabet} and $\mathcal B^*$ is set of \emph{codewords}.
\end{definition}
For our purposes, we will want binary codewords. So let us set $\mathcal B^* \equiv \{0,1\}^*$.
\begin{definition}
The code $E$ is \emph{non-singular} if it is injective.
\end{definition}
This permits unambiguous decoding since an injective encoding function $E$ has a well-defined decoding function $D= E^{-1}$.

In many situations, we want to encode sequences of source symbols, which result in a stream of concatenated codewords (think of a communication channel).

\begin{definition}
An \emph{extension} of a code $E$ defines a mapping 
$$
E^*: \mathcal A^* \to \mathcal B^*
$$
through the requirement that $E^*(s_1s_2 \ldots s_n)= E(s_1)E(s_2) \ldots E(s_n)$.
\end{definition}

\begin{definition}
A code~$E$ is \emph{uniquely decodable} if its extension is non-singular\,.
\end{definition}

This means that if one recieves $E(s_1)E(s_2) \ldots E(s_n)$, one can unambigously decoded it as $s_1 s_2 \ldots s_n$.

\begin{definition}
A code~$E$ is a \emph{prefix code} if $E(\mathcal A)$ is prefix-free, i.e. no codeword is the prefix of another codeword.
\end{definition}

\begin{example}\label{ex:codes}
Let $\mathcal A = \{A, B, C, D\}$ be source symbols to be encoded into bits. Consider the four codes defined in the table below.
\begin{center}
\begin{tabular}{c|c|c|c|c}
&$A$&$B$&$C$&$D$\\
\hline
$E_1$~ &10&10&11&0\\
\hline
$E_2$~ &10&110&1&0\\
\hline
$E_3$~ &0&01&011&111\\
\hline
$E_4$~ &0&10&110&111\\
\end{tabular}
\end{center}
\vspace{5pt}

\begin{itemize}[itemsep=1pt, topsep=2pt]
\item[--] $E_1$ is singular. The codeword~10 cannot be decoded.
\item[--] $E_2$ is not uniquely decodable.
How does one decode the stream of codewords~110? Is it $B$, $CA$, or $CCD$?. 

\item[--] $E_3$ is uniquely decodable but not a prefix code.

\item[--]  $E_4$ is a prefix code.
 \end{itemize}
 
 In a stream of codewords, prefix codes have the advantage of being \emph{instantaneously decodable}, i.e., a codeword can be decoded when its last bit is read. For this reason, prefix codes are also called \emph{instantaneous codes}.

Suppose that the sequence of codewords to decode is $0111010110$. \\
With the encoding~$E_4$, it can be instantaneously decoded into $ADABC$, by reading from left to right.
With the code~$E_3$, it can be decoded, but one needs to read further into the sequence. 
With~$E_2$ or~$E_1$, it cannot be decoded. See Fig.~\ref{fig:classesofcodes}.
\end{example}

Prefix sets~$\subset \{0,1 \}^*$ can be represented by trees. Fig.~\ref{fig:tree} represents the tree of the codewords of~$E_4$.

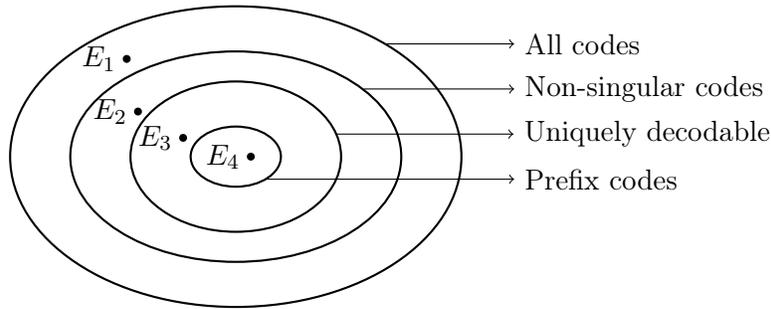
\begin{figure}
\centering
\begin{tikzpicture}[font=\small]

\def\RxA{3.0}  
\def\RyA{2.0}  
\def\RxB{2.2}
\def\RyB{1.4}
\def\RxC{1.4}
\def\RyC{1.0}
\def\RxD{0.6}
\def\RyD{0.4}

\draw[thick] (0,0) ellipse [x radius=\RxA, y radius=\RyA];

\draw[thick] (0,0) ellipse [x radius=\RxB, y radius=\RyB];

\draw[thick] (0,0) ellipse [x radius=\RxC, y radius=\RyC];

\draw[thick] (0,0) ellipse [x radius=\RxD, y radius=\RyD];


\draw[->] (\RxA-1.,1.5) -- ++(1.7,0) node[right]{All codes};

\draw[->] (\RxB-0.5,0.9) -- ++(2,0) node[right]{Non-singular codes};

\draw[->] (\RxC-0.07,0.3) -- ++(2.36,0) node[right]{Uniquely decodable};

\draw[->] (\RxD-0.2,-0.3) -- ++(3.3,0) node[right]{Prefix codes};


\fill (-1.45,1.3) circle(1.5pt) node[left]{$E_1$};

\fill (-1.3,0.6) circle(1.5pt) node[left]{$E_2$};

\fill (-0.7,0.25) circle(1.5pt) node[left]{$E_3$};

\fill (0.2,0) circle(1.5pt) node[left]{$E_4$};

\end{tikzpicture}
\caption{Classes of codes with the codes from Example~\ref{ex:codes}.}
\label{fig:classesofcodes}
\end{figure}

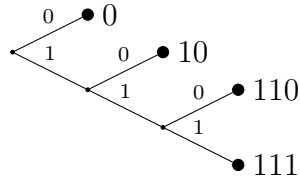
\begin{figure}
\centering
\begin{tikzpicture}
\node[circle, draw, fill=black, inner sep=0.5pt, name=root] at (0,0) {};

\node[circle, draw, fill=black, inner sep=1.5pt, name=n0] at (1,0.5) {};
\node[circle, draw, fill=black, inner sep=0.5pt, name=n1] at (1,-0.5) {};
\draw (root) -- node[midway, above] {\scriptsize 0} (n0);
\draw (root) -- node[midway, above] {\scriptsize 1} (n1);
\node[right=-1pt of n0] {$0$};

\node[circle, draw, fill=black, inner sep=1.5pt, name=n10] at (2,0) {};
\node[circle, draw, fill=black, inner sep=0.5pt, name=n11] at (2,-1) {};
\draw (n1) -- node[midway, above] {\scriptsize 0} (n10);
\draw (n1) -- node[midway, above] {\scriptsize 1} (n11);
\node[right=-1pt of n10] {$10$};

\node[circle, draw, fill=black, inner sep=1.5pt, name=n110] at (3,-0.5) {};
\node[circle, draw, fill=black, inner sep=1.5pt, name=n111] at (3,-1.5) {};
\draw (n11) -- node[midway, above] {\scriptsize 0} (n110);
\draw (n11) -- node[midway, above] {\scriptsize 1} (n111);
\node[right=-1pt of n110] {$110$};
\node[right=-1pt of n111] {$111$};

%
%

\end{tikzpicture}
\caption{The binary tree representing the codewords of $E_4$.}
\label{fig:tree}
\end{figure}

\nib{Observation.}
The lengths of the codewords of $E_4$ are $1,2,3,3$, and $$2^{-1}+2^{-2}+2^{-3}+2^{-3} = 1 \leq 1 \,.$$

\subsection{Kraft's Theorem}
\begin{theorem}[Kraft]
Let $l_1$, $l_2$, \ldots be a finite or infinite sequence of natural numbers. 
These numbers are the lengths of codewords of a prefix code if and only if
$$\sum_n 2^{-l_n} \leq 1\,.$$
\end{theorem}
\noindent This relation is called \emph{Kraft's inequality.}

\vspace{5pt}
\lecture{AIT-3}{https://youtu.be/JmUPHIDP2z8}
\begin{proof}

$\Rightarrow$.
Let~$\mathcal P$ be a set of codewords of a binary prefix code.
To each string $x \in \{0,1\}^*$ corresponds a specific kind of interval $\subset [0,1)$, called a \emph{cylinder}: 
\bes
x \mapsto \Gamma_x = [0.x, 0.x + 2^{-\len x})\,.
\ees
Above, $0.x$ denotes the real number in $[0,1)$ whose binary expansion begins with $x$.
 $\Gamma_x$ can be understood as the set of all real numbers whose first bits in their binary expansion are those of $x$.
By construction, considering any two binary strings, either one cylinder is contained in the other, or they are disjoint. The first case happens only in one string is the prefix of the other. The cylinders pertaining to a prefix-free set of codewords, like~$\mathcal P$, are, therefore, disjoint. Thus
$$
\sum_{x \in \mathcal P} 2^{-\len x} = \sum_{x \in \mathcal P} \text{lenght~}(\Gamma_x) \leq \text{length~}([0,1)) = 1 \,.
$$

\noindent $\Leftarrow$.
Without loss of generality, let $l_1$, $l_2$, ... be non-decreasing. Define
$$
\alpha_j = \sum_{n=1}^{j-1} 2^{-l_n} \qquad \text{and} \qquad \Gamma_{x_j} = [\alpha_j, \alpha_j + 2^{-l_j}) \,.
$$
The set of codewords given by $\{x_j\}$, where $0.x_j = \alpha_j$ and $\len{x_j} = l_j$, forms the codewords of the prefix code.

This assignment of codewords is better viewed on the binary tree. The codeword~$x_1$ is assigned to the lexicographically first node at depth~$l_1$. When a node is assigned, its descendant branches are cut out. Then~$x_2$ is assigned to the next available node at depth~$l_2$, and so on. 
\end{proof}

A code may have code-word lengths that fulfill Kraft's inequality without being a prefix code.
The code~$E_3$, in Example~\ref{ex:codes}, for instance. 
Yet, Kraft's theorem states that such code can be transformed into a prefix code while keeping the lengths of codewords fixed.
In the case of $E_3$, we reverse the bits of the codewords and obtain~$E_4$.

\begin{theorem}[McMillan]
The code-word lengths of any uniquely decodable code satisfy Kraft's inequality.
\end{theorem}
\begin{corollary}
The class of uniquely decodable codes, bigger than that of prefix codes, does not offer more admissible sets of code-word lengths.
Since no gain is obtained for optimality of code-word lengths if we look into this larger class, we can safely restrict our attention to prefix codes.
\end{corollary}

\subsection{Prefix codes of $\{0,1\}^*$}

Let us consider prefix codes of the alphabet $\mathcal A = \{0,1\}^*$.
\beas
E~:~\{0,1\}^* &\to& \{0,1\}^*\\
x &\mapsto& E(x) \,.
\eeas
The unary encoding with an end marker, 1, gives a prefix code:
\bes
E^{(0)}(x) = \underbrace{00\dots0}_{x~\text{times}}1 \equiv 0^{x}1 \qquad \text{(implicit identification between strings and $\mathbb N$).}
\ees
We can do better. Consider

\begin{minipage}{0.45\textwidth}
\beas
E^{(1)}(x) &=& E^{(0)}(x) x \\
&=& 0^{\len x} 1 x \equiv \bar x \\[-2pt]
\eeas
\end{minipage}
and
\begin{minipage}{0.45\textwidth}
\beas
E^{(2)}(x) &=& E^{(1)}(x) x \\
&=& 0^{\len{\len{x}}}1\len x x \,, \\[-2pt]
\eeas
\end{minipage}
where $\len{\len x}$ denotes the length of the length of $x$.
The encoding $E^{(1)}(x)$ will be often used. 
It is called \emph{the self-delimiting encoding} of $x$, and we denote it $\bar x$. 

The lengths of these codewords are
\beas
\len{E^{(1)}(x)} = 2 \len x + 1 = 2 \lfloor \log x \rfloor + 1 \qquad \text{and} \qquad \len{E^{(2)}(x)} = 2 \lfloor \log \lfloor \log x \rfloor \rfloor +\lfloor \log x \rfloor + 1  \,.
\eeas
It can be verified that these lengths satisfy Kraft's inequality. 

A prefix code $E: \{0,1\}^* \to \{0,1\}^*$ can be used to encode tuples or numbers into numbers: 
\beas
\mathbb N^k &\to& \mathbb N\\
(x_1,\dots, x_k) &\mapsto& E(x_1)\dots E(x_k) \,.
\eeas

\newpage
\section{Shannon Information Theory}
\lecture{AIT-4}{https://youtu.be/c518yTDsxmQ}

\begin{example}
Eight horses race. They have probabilities of winning given by 
$$(\pbfrac12,  \pbfrac14,  \pbfrac18, \pbfrac{1}{16},  \pbfrac{1}{64}, \pbfrac{1}{64}, \pbfrac{1}{64},  \pbfrac{1}{64} )\,.$$
At the end of a large number of races, we want to communicate in a binary (noiseless) channel the sequence of outcomes. How can we minimize the average number of bits to transmit? What is the expected number of bits to transmit per race?

Because the encodings will be concatenated, we need a uniquely decodable code, and without loss of generality, we shall take a prefix code.

The general idea is to save short codewords for probable symbols and use long codewords for rare symbols.

The proof of Kraft's theorem gives us a method.
The probabilities are
$$(2^{-1},  2^{-2},  2^{-3}, 2^{-4},  2^{-6}, 2^{-6}, 2^{-6}, 2^{-6} )\,.$$
And because they are probabilities, they sum to $1$.
Thus the numbers 
$$(1,  2,  3, 4,  6, 6, 6, 6 )$$
fulfil Kraft's inequality, and so
there exists a prefix code with codewords of these lengths\footnote{For instance, 
$\{0,  10,  110, 1110,  111100, 111101, 111110, 111111 \}$.}.
Since the lengths $l_x$ are related to the probabilities $p(x)$ via $p(x)= 2^{-l_x}$, the expected number of bits to transmit is
$$
\sum_{x=1}^8 p(x) l_x = \sum_{x=1}^8 p(x) \log \left(\frac{1}{p(x)} \right)\,.
$$
This expression shall be recognized as the Shannon entropy~$H(X)$ of the random variable~$X$ defined by the race. Here, $H(X) = 2$.
%
As we shall see, the entropy of a random variable is a lower bound on the expected length of a prefix code.
\end{example}

\subsection{Optimal Codes}
Let $X$ be a random variable over an alphabet $\mathcal X$.
What is the prefix code of minimal expected length?
In fact, what we care about is not the codewords themselves, but their lengths, i.e., the set $\{l_x\}_{x \in \mathcal X}$. This is because only the lengths determine the expected number of bits sent over the channel.
Thus we want to minimize 
$$
L = \sum_x p(x) l_x \,,
$$
with the constraint that $\sum_x 2^{-l_x} \leq 1 $ (because we want a prefix code).
Assuming an equality in the constraint
and omitting that $l_x \in \mathbb N$, it boils down to a standard minimization problem.
\bes
J= \sum_x p(x) l_x + \lambda \left( \sum_x 2^{-l_x} - 1 \right) \,.
\ees
\bes
\frac{\partial J}{\partial l_{x'}}=  p(x') - \lambda \ln 2 \cdot 2^{-l_{x'}} = 0 \qquad \iff \qquad  p(x') = \lambda \ln 2 \cdot 2^{-l_{x'}} \,.
\ees
Summing over $x'$ sets $\lambda \ln 2 = 1$ and we find that the minimal expected length is achieved with codewords of length
\bes
l^*_x = - \log p(x) = \log \left( \frac{1}{p(x)}\right)\,,
\ees
and the expected length is
\bes
L= \sum_{x\in \mathcal X} p(x) \log \left( \frac{1}{p(x)}\right) \deq H(X)\,.
\ees

\begin{definition}
The \emph{entropy} of a discrete random variable~$X$ is defined as
$$
H(X) \deq \sum_{x \in \mathcal X} p(x) \log \frac{1}{p(x)} \,.
$$
\end{definition}
The quantity $\log \frac{1}{p(x)}$, which grows as $p(x)$ gets smaller is sometimes called the \emph{surprise} or the \emph{self-information} of the event $\{X=x\}$. The entropy is thus the expectation of the surprise.

The operational meaning is clear from our derivation: it is the minimum expected number of bits to communicate the outcome of the random variable in a prefix code (or in a uniquely decodable code, thanks to McMillan's theorem).

\subsection{Properties of Entropy}
\lecture{AIT-5}{https://youtu.be/I9zIiq9dvao}
\vspace{5pt}
\begin{enumerate}[itemsep=1pt, topsep=2pt]
\item $H(X) \geq 0$; 
\item $H(X) = 0$ iff $X$ is deterministic;
\item If $|\mathcal X|$ is kept fixed, $H(X)$ is optimized by the equiprobable distribution, in which case $H(X) = \log |\mathcal X|$;
\item For equiprobable distributions, the larger~$|\mathcal X|$, the larger $H(X)$.
\end{enumerate}
\vspace{5pt}

Points 3 and 4 show that the more uncertain a random variable is, the larger its entropy.

\begin{definition}
The \emph{joint entropy} of a pair of discrete random variables $(X,Y)$ is defined as 
$$
H(X,Y) =  \sum_{x \in \mathcal X} \sum_{y \in \mathcal Y} p(x,y) \log \frac{1}{p(x,y)} \,.
$$
\end{definition}

If the random variable $X$ takes value $\hat x \in \mathcal X$, this defines $p_{Y \given X=\hat x}(y)$. Given each individual outcome $\hat x \in \mathcal X$ of $X$, there is residual entropy $H(Y \given X=\hat x)$ in the distribution $p_{Y \given X=\hat x}(y)$. The conditional entropy~$H(Y \given X)$ is the expectation (over $X$) of that residual entropy.
\begin{definition}
The \emph{conditional entropy}  is
\beas
H(Y \given X) &=&  \sum_{x \in \mathcal X} p(x) H(Y \given X=x) \\
&=& \sum_{x \in \mathcal X} p(x) \sum_{y \in \mathcal Y}  - p(y \given x) \log p(y \given x) \,.
\eeas
\end{definition}
\noindent $H(Y \given X) = 0$ iff the value of $Y$ is completely determined by the value of $X$.

\begin{theorem}[Chain rule] \label{crproba}
$$
H(X,Y) = H(X) + H(Y \given X) \,.
$$
\end{theorem}
\begin{proof}
In Problem Sheet \#1.
\end{proof}
\pscom{All notions for Problem Sheet \#1 have been covered. See \S\ref{sec:ps1}.}

\subsection{Mutual Information}\label{sec:mi}

\begin{definition}
The \emph{mutual information} between $X$ and $Y$ is
$$
I(X ; Y) \deq  H(X) - H(X \given Y)\,.
$$
It is the reduction in the uncertainty of $X$ due to the knowledge of $Y$. 
\end{definition}
See Fig.~\ref{fig:venn} for a representation of entropies, conditional entropies and mutual information.

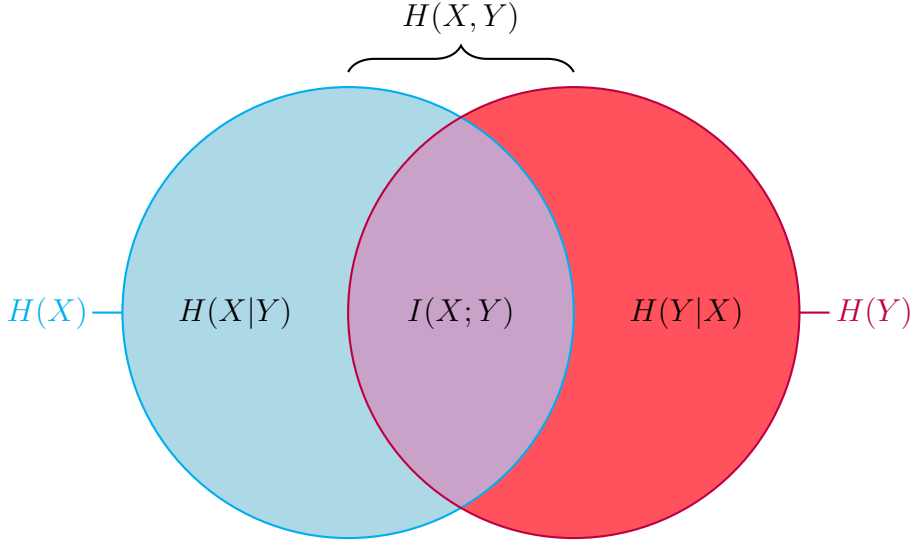
\begin{figure}
    \centering
    \begin{tikzpicture}
        \definecolor{lightblue}{RGB}{173,216,230}
        \definecolor{lightred}{RGB}{255,82,93}
        \definecolor{lightpurple}{RGB}{200,162,200}
        \definecolor{darkblue}{RGB}{0,0,139}
        \definecolor{darkred}{RGB}{139,0,0}

        \fill[lightblue] (-1.5,0) circle (3);
        \fill[lightred] (1.5,0) circle (3);
        \begin{scope}
            \clip (-1.5,0) circle (3);
            \fill[lightpurple] (1.5,0) circle (3);
        \end{scope}

        \draw[thick, cyan] (-1.5,0) circle (3);
        \draw[thick, purple] (1.5,0) circle (3);
        \draw[thick, cyan] (-4.5,0) -- (-4.9,0);
         \draw[thick, purple] (4.5,0) -- (4.9,0);

        \node at (-5.5,0) {\textcolor{cyan}{$H(X)$}};
        \node at (5.5,0) {\textcolor{purple}{$H(Y)$}};
        \node at (-3,0) {$H(X|Y)$};
        \node at (3,0) {$H(Y|X)$};
        \node at (0,0) {$I(X;Y)$};
        
    \draw [thick, decorate, decoration={brace, amplitude=10pt}] (-1.5,3.2) -- (1.5,3.2) 
    node[midway, above=10pt] {\( H(X,Y) \)};

    \end{tikzpicture}
    \caption{A Venn diagram representing the algebra of entropies, conditional entropies, and mutual information. The entropy~$H(X)$ is represented by the whole circular area enclosed by the blue line, while~$H(X \given Y)$ is represented by the sub-region outside the red circle. The mutual information~$I(X;Y)$ is represented by the overlapping region.}
    \label{fig:venn}
\end{figure}

The chain rule ensures that mutual information is symmetric:
\beas
I(X ; Y) &\stackrel{\text{c.r.}}{=}&  H(X) - \left(H(X , Y) - H(Y) \right) \\
&=& H(Y) - \left(H(X , Y) - H(X) \right) \\
&\stackrel{\text{c.r.}}{=}& H(Y) - H(X \given Y)\\
&=& I(Y ; X) \,.
\eeas

\nit{Information inequality.}
It can be shown that $I(X; Y) \geq 0$ and $= 0$ iff $X$ and $Y$ are independent. 
Thus $H(X) \geq H(X \given Y) $, so learning $Y$, on average, doesn’t make predicting~$X$ harder.

\begin{definition}
The random variables $X$, $Y$ and $Z$ form a \emph{Markov Chain}, denoted, $X \to Y \to Z$ if the joint distribution factorizes as
$$
p(x,y,z) = p(x)p(y\given x)p(z \given y)\,.
$$
\end{definition}
\noindent It means that $p(z \given y,x) = p(z \given y)$, i.e., $Z$ is conditionally independent of $X$, given $Y$.

\begin{theorem}[Data processing inequality]
If $X \to Y \to Z$, then $$I(X;Y) \geq I(X;Z)\,.$$
\end{theorem}
\begin{proof}
\beas
I(X ; Y, Z) &=& H(X) - H(X \given Y, Z)\\
&=& H(X) - H(X \given Z) + H(X \given Z) - H(X \given Y, Z)\\
&=& I(X ; Z) + I(X;Y \given Z) \,.
\eeas
and similarly,
\beas
I(X ; Y, Z) &=& H(X) - H(X \given Y, Z)\\
&=& H(X) - H(X \given Y) + H(X \given Y) - H(X \given Y, Z)\\
&=& I(X ; Y) + I(X;Z \given Y) \,.
\eeas
Since $Z$ is conditionally independent of $X$ given $Y$, $I(X;Z \given Y)=0$ and since $I(X;Y \given Z) \geq 0$, we find $I(X ; Z) \leq I(X ; Y) $.
\end{proof}

\begin{corollary}
In particular, if $Z = g(Y)$, we have $X \to Y \to g(Y)$, and so\\
 $I(X;Y)\geq I(X;g(Y))$.
\end{corollary}
It is impossible to increase the mutual information between two random variables by processing the outcomes in some deterministic manner---also in a probabilistic manner.

%


\newpage
\section{Computability Theory}
\lecture{AIT-6}{https://youtu.be/8gHkP2hmwJE}
\vspace{5pt}


Algorithms have existed for millennia — for instance, the Babylonians employed them as early as 2500 B.C.
Informally, they describe a procedure that can be performed to \emph{compute} a desired outcome.
In 1936, Church and Turing independently formalized \emph{computability} (and with it, algorithms).

%


\subsection{Turing's Model of Computation}\label{sec:TM}

Turing exhibited a simple type of hypothetical machine and argued that everything that can be reasonably said to be computed by a human computer using a fixed procedure can be computed by such a machine.

A Turing machine, like the one displayed in~Fig.~\ref{fig:TM}, consists of
\begin{itemize}[itemsep=1pt, topsep=2pt]
\item[--] An infinite \emph{tape} made of \emph{cells}, which can contain a $0$, a $1$, or a blank;
\item[--] A \emph{head}, capable of \emph{scanning} one cell;
\item[--] A \emph{finite control}, which can be in any one of a finite set of \emph{states} $Q$, and embody a \emph{list of rules};
\item[--] In the background there are discrete times $0,1,2, \ldots$
\end{itemize}

\begin{figure}[H]
\begin{tikzpicture}
    \draw (-2.5,0) -- (-2,0);
    \draw (-2.5,1) -- (-2,1);
    \draw (-2,0) rectangle (-1,1);
     \node at (-2.4,0.5) {\dots};
    
    \foreach \x/\val in {-1/, 0/0, 1/1, 2/, 3/0, 4/0, 5/1, 6/, 7/0, 8/1, 9/1} {
        \draw (\x,0) rectangle ++(1,1);
        \node at (\x+0.5,0.5) {\(\val\)};
    }
    
    \draw (10,0) rectangle (11,1);
    \draw (11,0) -- (11.5,0);
    \draw (11,1) -- (11.5,1);
    \node at (11.4,0.5) {\dots};

    \draw[thick, ->] (3.5,-0.75) -- (3.5,0);
    \node[above] at (2.9,-0.65) {head};

    \draw[thick] (2,-4) rectangle ++(3,3.25);
    
    \node at (3.5,-4.3) {finite control};

    \node[circle,draw,thick] at (2.6,-1.3) {\(q_0\)};
    \node[circle,draw,thick] at (2.6,-2.2) {\(q_1\)};
    \node[circle,draw,thick] at (2.6,-3.1) {\(q_2\)};
    \node at (2.6,-3.8) {\dots};
    \draw[thick, ->] (1.7,-2.2) -- (2.2,-2.2);
    \node at (0.4,-2.2) {current state};
    
    \node at (4.05,-1.3) {\scriptsize{\((q_0, 0, R, q_3)\)}};
    \node at (4.05,-1.7) {\scriptsize{\((q_0, 1, 0, q_0)\)}};
    \node at (4.05,-2.1) {\scriptsize{\((q_1, 0, L, q_2)\)}};
    \node at (4.05,-2.5) {\scriptsize{\((q_1, 1, L, q_2)\)}};
    \node at (4.05,-2.9) {\scriptsize{\((q_2, 0, 1, q_6)\)}};
    \node at (4.05,-3.3) {\scriptsize{\((q_2, 1, B, q_3)\)}};
    \node at (4.05,-3.8) {\dots};
    
     \draw [decorate,decoration={brace,amplitude=10pt},xshift=0pt,yshift=0pt]
    (4.9,-1) -- (4.9,-3.8) node [midway,xshift=40pt] {list of rules};

\end{tikzpicture}
\caption{A generic configuration of a Turing Machine.}
\label{fig:TM}
\end{figure}
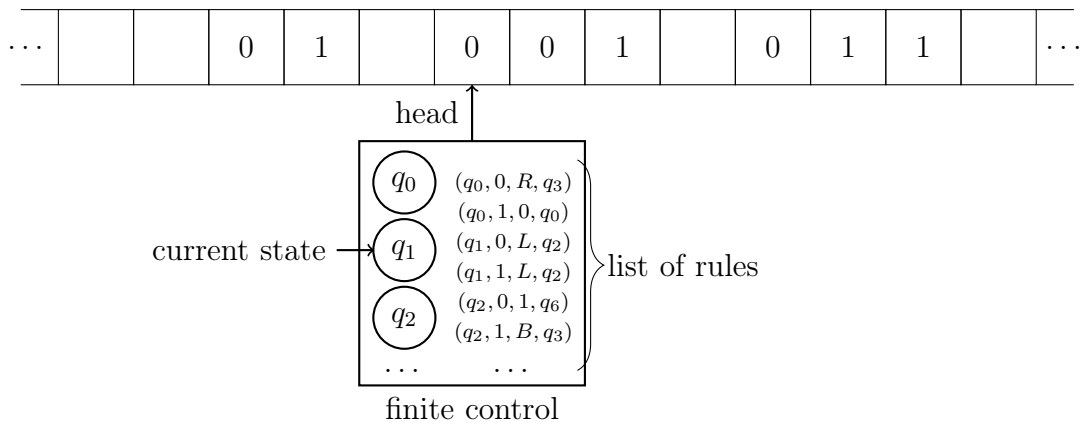

A computation starts at time $0$ with the head on a distinguished \emph{start cell}, and the finite control is in a distinguished \emph{start state} $q_0$. Initially, all cells of the tape are blank, except for a contiguous finite sequence extending from the start cell to the right, which contains 0s and 1s. This binary sequence is called the \emph{input}. 

The Turing machine performs one \emph{operation} per time step in which it does one of the five following actions:
writes a $0$ in the cell that is being scanned, 
writes a~$1$, 
leaves a blank, 
moves the head to the left (L) or 
moves the head to the right (R).
The action performed depends on the current state~$q$ and the scanned symbol. At the end of each operation, the control takes a (possibly different) state~$q' \in Q$. 

For a given Turing machine, its set of operations is given by the finite list of rules of the form $(q, s, a, q')$, where $q, q' \in Q$, $s \in \{0,1,B\}$  and $a \in\{0,1,B,L,R\}$ ($B$ stands for ``blank'').
This denotes the state~$q$ and the scanned symbol~$s$ before the operation, the action~$a$ taken, and the state~$q'$ after the operation. 

 
The machine is \emph{deterministic}: $(q, s, a_1, q'_1)$  and $(q, s, a_2, q'_2)$ cannot both be rules if $(a_1, q'_1) \neq (a_2, q'_2)$.
Moreover, not every possible combination of the first two elements has to be in the list: if the machine enters a state~$q$ and reads a symbol~$s$ for which no rule is defined, it \emph{halts}. 
Hence, we can define a Turing machine by a mapping from a finite subset of $Q \times \{0,1,B\} $ into $ \{0,1,B,L,R\} \times Q$. 

A given Turing machine, on a given input, carries out a uniquely determined succession of operations, which may or may not terminate in a finite number of steps. 

\subsection{Mathematizing the mathematician}

Is Turing's 1936 paper~\cite{Turing1936} a piece of mathematics or a piece of physics?
As is manifest in his \S9's ``direct appeal to intuition'', he explains that his model correctly describes a specific physical system, namely, a mathematician. 
This kind of activity belongs to physics. 

He gives arguments (no proof) that his conjectured description is faithful.
For instance, Turing explains why he uses a one-dimensional tape made of squares:
\begin{quote}
``Computing is normally done by writing certain symbols on paper. We may suppose this paper is divided into squares like a child's arithmetic book. In elementary arithmetic the two-dimensional character of the paper is sometimes used. But such a use is always avoidable, and I think that it will be agreed that the two-dimensional character of paper is no essential of computation. I assume then that the computation is carried out on one-dimensional paper, i.e. on a tape divided into squares.''
\end{quote}

He explains why only a finite alphabet is admissible:
\begin{quote}
 ``I shall also suppose that the number of symbols which may be printed is finite. If we were to allow an infinity of symbols, then there would be symbols differing to an arbitrarily small extent.''
\end{quote}

Further, Turing writes:
\begin{quote}
``The behaviour of the computer at any moment is determined by the symbols which he is observing, and his `state of mind' at that moment.'' 
\end{quote}
The ``computer'' he is referring to cannot be what we call a computer, for it is in that very paper that Turing is laying down the foundational work that gave us our computers. So what is he referring to? A ``computer'' was how a \emph{person} carrying out a computation was called. 

Turing continues, arguing that the range of squares that a computer can read is limited:
\begin{quote}
``We may suppose that there is a bound $B$ to the number of symbols or squares which the computer can observe at one moment. If he wishes to observe more, he must use successive observations.''
\end{quote} 
Yet it can shown that a machine scanning one single symbol, like the ones we defined in \S\ref{sec:TM}, can simulate the computer that is scanning $B$ symbols at a time. 

Turing then explains why there are only finitely many states (here ``states of mind'').
\begin{quote}
``We will also suppose that the number of states of mind which need be taken into account is finite. The reasons for this are of the same character as those which restrict the number of symbols. If we admitted an infinity of states of mind, some of them will be ``arbitrarily close'' and will be confused. Again, the restriction is not one which seriously affects computation, since the use of more complicated states of mind can be avoided by writing more symbols on the tape.''
\end{quote}

\subsection{Partial Functions}
\lecture{AIT-7}{https://youtu.be/PGw33WexToM}
\vspace{5pt}

A function $f: A \to B$ is \emph{partial} if its domain of definition $D$ is not necessarily equal to $A$: there may be some $a \in A$ for which $f(a)$ is undefined. 
When $D=A$, $f$ is \emph{total}.

We can associate a partial function $f_{T}: \{0,1\}^* \to \{0,1\}^*$  (or $\mathbb N \to \mathbb N$) to each Turing machine $T$.
The input of $T$ is understood as the input of $f_{T}$.
If on input $x$, $T$ halts, the string of $0$s and $1$s that is being scanned (and delimited by blanks) defines $f_{T}(x)$.
If on input $x$, $T$ does not halt, $f_{T}(x)$ is undefined.

The function $f: \mathbb N \to \mathbb N$ is \emph{computable}\footnote{Also called \emph{recursive}; the adjective \emph{effective} is also commonly used in this context.} if there exists a Turing machine $T$ for which $f = f_{T}$.

%
It is \emph{partial computable} or \emph{total computable} if $f$ is respectively partial or total. 
When the context is clear, we identify the function~$f_T$ computed by~$T$ with~$T$ (so we may write~$T(x)$ for~$f_T(x)$).

Although the set of functions $\mathbb N \to \mathbb N$ appears restrained, via suitable encodings of other discrete sets into $\mathbb N$, the computable functions may be viewed as computing other kinds of functions of a countable set into another.

\begin{example}
A function 
\beas
g : \mathbb N \times \ldots \times \mathbb N &\to& \mathbb N \\
(x_1, \ldots, x_n) & \mapsto & g(x_1, \ldots, x_n)
\eeas
can be viewed as being computed by a Turing machine expecting an input of the form $\bar x_1\bar x_2 \ldots \bar x_n$.

\end{example}



\subsection{Effective Enumeration of Turing Machines}

The list of rules of a Turing machine $T$ can be encoded in a string.
Fixing an encoding $e$ of $Q \cup \{0,1,B,L,R\}$ into $s = \lceil \log (|Q| + 5) \rceil$ bits, each rule can be encoded as $e(q)e(s)e(a)e(q')$.
Let $r$ be the number of rules; $r \leq 3 |Q|$, because each state of a Turing machine can have at most one rule defined per possible scanned symbol.
The list of rules of $T$ can thus be encoded as
$$
E(T) = \bar s \ \bar r \ e(q_1)e(s_1)e(a_1)e(q'_1) e(q_2)e(s_2)e(a_2)e(q'_2) \ldots e(q_r)e(s_r)e(a_r)e(q'_r) \,.
$$
(Remember that~$\bar s = 1^{\len s}0s$ is the self-delimiting encoding of $s$.)
The list of rules encodes all that there is to know about $T$.  The set of states is contained in the list of rules and the initial state may be taken to be the state in the first rule.

All bit strings that correspond to a valid encoding of some Turing machine can be listed lexicographically and, in turn, numbered.
The \emph{index} (or \emph{Gödel number}) of a Turing machine~$T$ is~$i$ if $E(T)$ is the $i$\textsuperscript{th} element in the lexicographic order of valid encodings of Turing machines.
This yields a sequence of Turing machines $T_1$, $T_2$, $\ldots$ called an \emph{effective enumeration}. 

\begin{remark}
It is \emph{effective}, in the sense that one can construct a Turing machine~$M$ which, on input $i$, returns $E(T_i)$. 

\begin{center}
\begin{tabular}{rl}
\smallskip
$M$:~ &\texttt{On input $i$,}\\
&\texttt{Initialize a counter $S=0$.}\\
&\texttt{For each $y \in \{0,1\}^*$ (in lexicographical order),}\\
&\texttt{\qquad check if $y$ corresponds to $E(T)$, for some $T$.}\\
&\texttt{\qquad If it does, add $1$ to a counter $S$.}\\
&\texttt{\qquad If $S = i$, output $y = E(T_i)$.}
\end{tabular}
\end{center}
\end{remark}

\subsection{The Universal Turing Machine}

There are infinitely many Turing machines. 
But one of them can simulate them all.

\begin{definition}
A \emph{universal Turing machine} $U$ is a machine that expects an input encoding a pair $(i, x)$ and simulates the machine $T_i$ on input $x$. Therefore,
$$
U(\langle i, x \rangle) = T_i(x) \qquad \forall i \forall x \,.
$$
\end{definition}

\lecture{AIT-8}{https://youtu.be/aGeZF5Ctfsk}\\
\nit{Comment on the lecture: 
There is a brief discussion of Problem Sheet \#1,~\S\ref{sec:ps1}.}

\begin{theorem}
A universal Turing machine exists and can be constructed effectively. 
\end{theorem}
\begin{proof}[Proof sketch]
$U$ invokes $M$ so that, from $i$, it reconstructs the list of rules $E(T_i)$, which it saves on a portion of its tape.
To execute the consecutive actions that $T_i$ would perform (on input $x$) on its own tape, $U$ uses $T_i$’s rules to simulate $T_i$’s actions on a representation of $T_i$’s tape contents. 
\end{proof}
In fact, there are infinitely many such $U$’s.

\subsection{Church--Turing--Deutsch}


Why do Turing machines matter so much? Why not some other model of computation? 
The answer lies in the robustness of the notion of computability itself---which we now explore through the Church–Turing thesis, and its physical extension by Deutsch.


It can be proved that the set of functions $\mathbb N \to \mathbb N$ that can be computed by Turing machines is equal to the set of functions that can be computed by many other systems of symbolic manipulations:

\vspace{8pt}
\begin{minipage}{0.45\textwidth}
\begin{itemize}[itemsep=0.9pt, topsep=5pt]
\item[--] $\lambda$-calculus 
\item[--] General recursive functions
\item[--] Python, C, Lisp, Java,...  \vspace{16pt}
\end{itemize}
\end{minipage}
\begin{minipage}{0.45\textwidth}
\begin{itemize}[itemsep=0.9pt, topsep=3pt]
\item[--] TeX
\item[--] Conway's game of life
\item[--] Minecraft
\item[--] \dots
\end{itemize}
\end{minipage}
\vspace{8pt}



\nib{The Church--Turing Thesis}

\ct{There is an objective notion of computability independent of a particular formalization---and it captures what is computable by a human.} 


Turing machines are a (convenient) representative of this notion of computation.

\begin{definition}
A system of symbolic manipulation (such as a computer's instruction set, a programming language, or a cellular automaton) is said to be \emph{Turing-complete} or \emph{computationally universal} if it can be used to simulate any Turing machine. It is sufficient to show that it can simulate the universal Turing machine.
\end{definition}

Is computational universality a property of mathematics, metamathematics or our physical universe? David Deutsch wrote~\cite[Chapter 5]{deutsch1997fabric}:

\begin{quote}
Computers are physical objects, and computations are physical processes. What computers can or cannot compute is determined by the laws of physics alone, and not by pure mathematics. 
\end{quote}

The argument is short yet inescapable. 
Qua theoretical construct, Turing machines may appear to be an abstract piece of mathematics. 
But so do pseudo-Riemannian manifolds, which are conjectured to describe spacetime in general relativity.
Thus, whether Turing machines are the relevant theoretical construct depends on whether they actually grasp the set of computations that can be carried out physically.

\begin{example}[The Zeno machine]
Consider a regular Turing machine with an extra functionality: a big step. In one single big step, the Zeno machine can perform infinitely many operations of a regular Turing machine, as if it were doing one operation in $\pbfrac 12$ time unit, another operation in $\pbfrac 14$ time unit, another in $\pbfrac 18$ time, and so on.
It can be proven that the Zeno machines compute more functions than the Turing machines.
Based on what do we reject Zeno machines as a valid formalization of computation?
\emph{Physics. Try and build that machine...}
\end{example}

When computer science is incorporated into a physical worldview, the Church--Turing thesis can be seen as a special case of a physical principle due to Deutsch.\\

\nib{The Church--Turing--Deutsch Principle}

\ct{It is possible to build a universal computer: a machine that can simulate (to arbitrary accuracy) any physical process.}

Since computations are special kinds of physical processes, a universal computer can perform any computation.
%
Moreover, the sense in which the ``objective notion of computability'' is objective in the Church--Turing thesis is clarified: it is objective insofar as it pertains to the physical world.
Humans are computationally universal.

\begin{example}
Suppose that in the year  \underline{~~~~}25, a new theory of physics is discovered:  physical systems can undergo different processes simultaneously.
Since computations are physical processes, this functionality can be harnessed, perhaps to speed up computations due to parallelism.
\emph{The year is 1925, this theory is quantum theory, and the speed up is that of quantum computation.}
\end{example}

By understanding computation physically, Deutsch came up with the universal quantum computer~\cite{deutsch1985quantum}.

It should be noted that quantum computers can be simulated by Turing machines with an exponential overhead in computation time. Since in this course we are interested in what can or cannot be computed---and not by computation time---we can safely restrict our attention to the universal Turing machine.


\subsection{Computing, Deciding, Enumerating, Recognizing}

\lecture{AIT-9}{https://youtu.be/BCJGenVScHo}\\
\nit{Comment on the lecture: Some solutions to Problem Sheet \#1 (\S\ref{sec:ps1}) are presented.}
\vspace{5pt}
\vspace{-8pt}

The effective enumeration of Turing machines $T_1, T_2, T_3, \ldots$ provides an enumeration of all partial computable functions $\varphi_1, \varphi_2, \varphi_3, \ldots$. We shall denote\\[2pt]
-- $\varphi_i(j) \halts$ ~(and $T_i(j) \halts$)~  if $\varphi_i(j)$ is defined, i.e., if $T_i$ halts on input $j$ and\\[2pt]
-- $\varphi_i(j) \nohalt$ ~(and $T_i(j) \nohalt$)~ if $\varphi_i(j)$ is undefined, i.e., if $T_i$ does not halt on input $j$.

\vspace{2pt}
Let $\alpha \in \mathbb R$. Let $\finite{\alpha}{n}$ denote the first $n$ bits of $\alpha$ in binary expansion. A number $\alpha \in \mathbb R$ is \emph{computable} if the function $n \mapsto \finite{\alpha}{n} $ is computable.

\begin{definition}
Let $A\subseteq \mathbb N$.
A \emph{decider}~$D_A$ of the set $A$ is a Turing machine that computes the characteristic function of $A$,
$$
D_A(n) = \begin{cases} 1 &\text{ if } n \in A \\
0 &\text{ if } n \notin A.
\end{cases}
$$
It decides if, yes or no, a given element is in $A$.
\end{definition}

\begin{definition}
A set is \emph{computable} (or \emph{decidable} or \emph{recursive}) it has a decider.
\end{definition}

\begin{example}
All finite sets are decidable. The algorithm may just include the list of all elements and verify.
The set of odd numbers, the set of prime numbers, and the set of powers of $17$ are decidable.
\end{example}

\begin{definition}
Let $A\subseteq \mathbb N$. An \emph{enumerator} $E_A$ of the set $A$ is a Turing machine that enumerates all (and only) elements of $A$, without order and possibly with repetitions. In general, the enumeration is not a halting computation, as must be the case if~$A$ is infinite. 
But even if $A$ is finite, the enumerator might forever be looking for more elements of $A$.
\end{definition}

\begin{definition}
A set $A\subseteq \mathbb N$ is \emph{computably enumerable} (or \emph{recursively enumerable}) if it has an enumerator.
\end{definition}

\begin{example}\label{ex:re} 
The following sets are computably enumerable.
\begin{enumerate}[itemsep=0.9pt, topsep=5pt]
\item All decidable sets.
This is because the decider can be transformed into an enumerator. Run $D_A(1), D_A(2), \ldots$ if $D_A(a)=1$ add $a$ to the enumeration.
\item The set 
$$\{\langle n, m \rangle | n \text{~is the largest number of consecutive 0s in the first $2^m$ digits of~} \pi= 3.14159 \ldots \}.$$
\item The set of indices $i$ such that the range of $\varphi_i$ is nonempty:
$
\{i | \exists j \st \varphi_i(j) \halts \}
$.
Indeed, consider the following enumerator
\vspace{-8pt}
\begin{center}
\begin{tabular}{rl}
$E$:~ &\texttt{For $k= 1, 2, 3, \ldots$}\\
&\texttt{\qquad simulate all $\{T_i(j)\}_{i,j \leq k}$ for $k$ steps of computation.}\\
&\texttt{\qquad If some $T_i(j) \halts$ (which means $\varphi_i$ is defined on $j$),}\\
&\texttt{\qquad \qquad output $i$}
\end{tabular}
\end{center}
\vspace{-4pt}
Notice that if $\varphi_i$ has a non-empty range, there is some~$j$ and some~$t$ such that the computation $T_i(j)$ halts after $t$ steps. Eventually (when~$k= \max(i,j,t)$), the enumerator will find it out and output $i$ in the enumeration. 
On the other hand, if~$\varphi_i \nohalt$ for all inputs, the enumerator will never output~$i$. 

\end{enumerate}

\end{example}

\begin{definition}
A set $A \subseteq \mathbb N$ is \emph{recognizable} (or \emph{semi-decidable}) if the following partial computable function exists 
$$
R_A(n) = \begin{cases} 1 &\text{ if } n \in A \\
0 \text{ or } \nohalt &\text{ if } n \notin A.
\end{cases}
$$
\end{definition}

%

\begin{proposition}
A set $A \subseteq \mathbb N$ is computably enumerable $\iff$ it is recognizable.
\end{proposition}
\begin{proof}
$\Longrightarrow$. $E_A \to R_A$. On a given $n \in \mathbb N$, run the enumerator $E_A$ until $n$ is outputted, and then accept it (map $n$ to $1$). Let's verify it works.\\
If $n \in A$, it will eventually be enumerated by $E_A$, and so accepted.\\
If $n \notin A$, $n$ will never be enumerated, and the vain search will loop, as required.\\
\\
$\Longleftarrow$. $R_A \to E_A$. For $k = 1, 2, 3, \ldots$, run $\{R_A(n)\}_{n \leq k}$ for $k$ steps of computation. If some $n$ is found for which $R_A(n)= 1$, it means $n \in A$, enumerate that $n$. Note that this procedure will never enumerate some $n \notin A$.
\end{proof}

\subsection{The Halting Problem}
\lecture{AIT-10}{https://youtu.be/N6vit1y5j8o}

\begin{definition}
The \emph{halting set} is defined as
$$K_0 \deq \{\langle i,j \rangle | \varphi_i(j) \halts\} \,.$$
\end{definition}

\begin{proposition}
$K_0$ is computably enumerable.
\end{proposition}
\begin{proof}
In Problem Sheet \#2 (\S\ref{sec:ps2}). The proof is similar to item 3. of Example~\ref{ex:re}.
\end{proof}

\begin{theorem}
The halting set $K_0$ is undecidable.
\end{theorem}
The proof uses the diagonalization method invented by Cantor when he proved that there are uncountably many real numbers. He assumed that the real numbers could be enumerated and reached a contradiction by exhibiting a real number that cannot be on the list. The contradiction forces us to abandon the possibility of enumerating the real numbers. Here, the Turing machines and their inputs really can be enumerated. The assumption that leads to the contradiction is that the table~$i,j \mapsto T_i(j)$ can be computed. See the proof and Table~\ref{fig:halting}.

\begin{proof}
Suppose that a Turing machine $D_{K_0}$ decides $K_0$.
$$
D_{K_0} (\langle i, j \rangle)= 
\begin{cases} 
1 \qquad &\text{ if } T_i(j) \halts \\
0 \qquad &\text{ if } T_i(j) \nohalt.
\end{cases}
$$
The machine $D_{K_0}$ can be used to construct the machine 
$$
\Delta (i)= 
\begin{cases} 
1 &\text{ if } T_i(i) \halts \\
0 \qquad &\text{ if } T_i(i) \nohalt \,,
\end{cases}
$$
which in turn can be transformed into the machine
$$
\bar \Delta (i)= 
\begin{cases} 
\nohalt &\text{ if } T_i(i) \halts \\
0 \qquad &\text{ if } T_i(i) \nohalt \,.
\end{cases}
$$
Like all Turing machine, $\bar \Delta$ is in the effective enumeration, so~$\bar \Delta = T_\delta$ for some~$\delta \in \mathbb N$. What does~$\bar \Delta$ do on input~$\delta$?
$$
\bar \Delta (\delta)= 
\begin{cases} 
\nohalt &\text{ if } T_\delta(\delta) \halts \\
0 \qquad &\text{ if } T_\delta(\delta) \nohalt \,,
\end{cases}
$$
so 
$$
\bar \Delta (\delta) \nohalt \text{ if } \bar \Delta(\delta) \halts 
\qquad \text{and} \qquad
\bar \Delta (\delta) \halts \text{ if } \bar \Delta(\delta) \nohalt \,,
$$
a contradiction. See Table~\ref{fig:halting}.
\end{proof}

\begin{table}
\centering
\begin{tabular}{c|ccccccc}
  & 1 & 2 & 3 & 4 & 5 & \dots & $\delta$ \\
\hline
$T_{1}$      & $\boxed 3$ & $\nohalt$ & 3 & 1 & $\nohalt$ & \dots &  \\
$T_{2}$      & 1 & $ \boxed \nohalt$ & 1 & $\nohalt$ & $\nohalt$ & \dots &  \\
$T_{3}$      & 2 & 2 & $ \boxed \nohalt$ & 2 & $\nohalt$ & \dots &  \\
$T_{4}$      & 1 & $\nohalt$ & 1 &  \boxed 2 & $\nohalt$ & \dots & \\
$T_{5}$      & 7 & 7 & 1 & $\nohalt$ &  \boxed 7 & \dots & \\
$\vdots$     & $\vdots$ & $\vdots$ & $\vdots$ & $\vdots$ & $\ddots$ & $\vdots$ \\
$T_{\delta}$ & $\nohalt$ & 0 & 0 & $\nohalt$ & $\nohalt$ & \dots & {\color{red}  \boxed ?} \\
\end{tabular}
\caption{Table representing the results of all computations. The machine $T_\delta$ is constructed by inverting the behaviour of the diagonal of the table.}
\label{fig:halting}
\end{table}

Thus there exists no fixed algorithm which decides if a given Turing machine~$i$ halts on a given input~$j$.
If~$T_i(j) \halts$, this can be verified: just run it.
If~$T_i(j) \nohalt$, running it is a bad idea---no amount of waiting can confirm non-halting.
If this is transposed into modern programming languages, how can one tell if a program will loop?
For trivial loops, it can be recognized.
But what the halting problem tells us is that there are ever more convoluted ways in which a program can loop so that no fixed algorithm can capture all those ways.

The halting set $K_0$ can be embedded into a real number $ \chi = \chi_1 \chi_2 \chi_3 \ldots $, where
$$
\chi_n = \begin{cases}
 1 & \text{if } n \in K_0 \\
 0 & \text{if } n \notin K_0 \,.
\end{cases}
$$
This $\chi$ is Turing's \emph{uncomputable number}.


\subsection{Uncomputability $\boldsymbol \implies$ Incompleteness}

We’ve now seen that some functions cannot be computed---even in principle. But what about proving properties of such functions? Can a sufficiently clever mathematician---or a sufficiently powerful formal system—sidestep uncomputability by proving whether a program halts?
Because a proof is itself a kind of computation, uncomputability spills over into provability—and we now enter the realm of incompleteness.

\begin{definition}
A \emph{formal axiomatic theory} (\fat) is defined by
\begin{itemize}[itemsep=0.9pt, topsep=5pt]
\item[--] an alphabet
\item[--] a grammar
\item[--] a set of axioms
\item[--] a set of rules of inference (logic) and 
\item[--] a proof-checking algorithm. 
\end{itemize}
In Hilbert's spirit, proofs must be made so precise that a machine can verify their validity.
In a \fat, a \emph{theorem} is a statement that can be derived from the axioms and the rules of inference. The \emph{proof} is the derivation.
\end{definition}

The defining properties of \fats\ imply that the theorems of any given \fat\ can be computably enumerated by a Turing machine. 
Such a Turing machine embodies the \fat (the axioms and rules of inference can be encoded finitely), and enumerates, one by one, all theorems in order of proof size (not in order of theorem-statement size). How? Two ways.

\begin{enumerate}[itemsep=0.9pt, topsep=5pt]
\item For $i= 1,2,3, \ldots $, output all statements obtained by $i$ applications of a rule of inference to the axioms.
\item Working through all strings $p$ (in length increasing fashion), verify if $p$ is a proof of some statement. If yes output the statement.
\end{enumerate}

In computer-theoretic terms, a formal axiomatic theory is a computably enumerable set of mathematical assertions. 

\begin{theorem}[Incompleteness \`a la Turing]
For all \fats, there exists some $i$ and $j$ for which $T_i(j) \nohalt$ is true but unprovable.
\end{theorem}
\begin{proof}
By contradiction: Suppose that there exists a \fat\ $\mathcal T$ in which for all $\langle i,j\rangle$ such that $T_i(j) \nohalt$, it can be proven that indeed $T_i(j) \nohalt$.
A decider for the halting problem can be exhibited by executing two computations in parallel:
\vspace{-6pt}
\begin{center}
\begin{tabular}{rl}
$D$:~ &\texttt{On input $\langle i,j\rangle$,}\\
&\texttt{For $t=1, 2, \dots$}\\
&\texttt{\quad Run $T_i(j)$ for $t$ steps} (looking for $T_i(j) \halts$)\\
&\texttt{\quad Run the enumerator of $\mathcal T$ for $t$ steps} (looking for a proof of $T_i(j)\nohalt$)\\
&\quad (One of these computations will eventually find the halting status of $T_i(j)$.)\\
&\quad \texttt{Output that status.}\\
\end{tabular}
\end{center}

But no such decider can exist.
\end{proof}
We will return to incompleteness in~\S\ref{sec:chaitinincomplete} when we discuss Chaitin's incompleteness theorem.

\subsection{Rice's Theorem}
\lecture{AIT-11}{https://youtu.be/DvpGnTEPej4}
\vspace{5pt}

What has been presented so far is sufficient computably theory to move to algorithmic information. However, it is worth presenting Rice's theorem, because it shows that what appears to be a specific uncomputable problem, the halting problem, contaminates many, many other problems.
In a nutshell, Rice's theorem states that any non-trivial property concerning the behaviour of Turing machines (or computer programs) is undecidable.

It is perhaps more enlightening to begin with an example and later generalize it into the theorem.
\begin{example}
Consider the successor function 
\beas
S: \mathbb N &\to& \mathbb N \\
n & \mapsto & n+1 \,.
\eeas
In the enumeration $\varphi_1, \varphi_2, \varphi_3 \ldots$ of all partial computable functions, the function $S$ occurs infinitely many times.
Consider
$$
I_S = \{i \st \varphi_i = S\}
$$
\nit{Problem:} Is $I_S$ decidable? I.e. does there exist a Turing machine that computes
$$
D_{I_S} (i) = \begin{cases}
1 & \text{if } \varphi_i = S \\
0 & \text{if } \varphi_i \neq S
\end{cases} \qquad ?
$$
In modern terms: can a fixed algorithm generically decide if a given program computes the successor function?

\smallskip
\nit{Answer:} No, the set~$I_S$ is undecidable.\\
\smallskip
\nit{Proof.}
Suppose by contradiction that $D_{I_S}$ exists.
%
First, construct a Turing machine~$T_k$ that computes~$S$. 
(The decider~$D_{I_S}$ can be used on input~$k$ to verify that~$T_k$ does indeed compute $S$.)
Then, from~$T_k$, we construct a family of Turing machines~$\{ T_k^{(i,j)}\}_{i,j \in \mathbb N}$, with the parameters~$i$ and~$j$ hardcoded in the machine~$T_k^{(i,j)}$.
\begin{center}
\begin{tabular}{rl}
 $T^{(i,j)}_k$~: &\texttt{ On input $n$}\\
&\texttt{\qquad Run $T_i(j)$}\\
&\texttt{\qquad if it halts}\\
&\texttt{\qquad \qquad Run and output $T_k(n)$.}\\
\end{tabular}
\end{center}

\noindent Note that if $T_i(j) \halts$, then $T^{(i,j)}_k$ does compute $S$ and\\
\hphantom{Note that }if $T_i(j) \nohalt$, then $T^{(i,j)}_k$ does not compute $S$.

\noindent Note also that 
$$
(i,j) \underbrace{\to T^{i,j}_k \to E(T^{i,j}_k) \to}_{\text{Turing machine } M} \text{index of }T^{i,j}_k
$$
is computable because each stage is computable. Therefore, 
$$
\langle i,j \rangle \stackrel{M}{\to} \text{index of }T^{i,j}_k \stackrel{D_{I_S}}{\to}  \begin{cases}
1 & \text{if } T_i(j) \halts \\
0 & \text{if } T_i(j) \nohalt
\end{cases} 
$$
would be a decider for $K_0$. \qed
\end{example}


\begin{definition}
An \emph{index set} $I$ is a set with the following property:
If $\varphi_i = \varphi_j$ then 
$$
i \in I \iff j \in I
$$
\end{definition}
Index sets are a way to describe classes of partial computable functions or properties of partial computable functions. In the example, $I_S$ is an index set.
\begin{definition}
An index set $I$ is \emph{trivial} if either $I = \emptyset$ or $I = \mathbb N$.
\end{definition}

\begin{theorem}[Rice]
Any non-trivial index set is uncomputable.
\end{theorem}
\begin{proof}
Let $I$ be a non-trivial set of indices.
Suppose (by contradiction) that $D_{I}$ exists.
Program your favourite Turing machine $T_l$ that loops on all inputs, and compute its index $l$.
Compute $D_{I}(l)$.

\nib{Case 1:} If $D_{I}(l) = 0$, i.e., $l \not\in I$,
compute $D_{I}(1)$, $D_{I}(2)$,... until a $k \in I$ is found.

\nib{Case 2:} If $D_{I}(l) = 1$, i.e., $l \in I$,
compute $D_{I}(1)$, $D_{I}(2)$,...  until a $k \not\in I $ is found.

By construction,~$T_{k}$ computes a function that does not loop on all inputs---i.e., it halts on at least one input.
We can then define a family of Turing machines~$\{ T_{k}^{(i,j)}\}_{i,j \in \mathbb N}$ that have~$i$ and~$j$ harded coded into them as follows:
\begin{center}
\begin{tabular}{rl}
$T^{(i,j)}_{k}$: & \texttt{On input $n$,}\\
&\texttt{Run $T_i(j)$}\\
&\texttt{if it halts}\\
&\texttt{\qquad Run $T_{k}(n)$}.
\end{tabular}
\end{center}

Note that if $T_i(j) \halts$, then $T^{(i,j)}_{k}$ has the same behaviour as $T_{k}$. Therefore, $k$ and the index of $T^{(i,j)}_{k}$ are either both elements of $I$ (case 1) or both not elements of $I$ (case 2). \\
If $T_i(j) \nohalt$, then $T^{(i,j)}_{k}$ loops on all inputs, therefore either $k \in I$ and the index of $T^{(i,j)}_{k}$ is not in $I$ (case 1) or vice versa (case 2).

\noindent Note also that 
$$
(i,j) \underbrace{\to T^{i,j}_k \to E(T^{i,j}_k) \to}_{\text{Turing machine } M} \text{index of }T^{i,j}_k
$$
is computable.
Therefore, in case 1,
$$
\langle i,j \rangle \stackrel{M}{\to} \text{index of }T^{i,j}_k \stackrel{D_{I}}{\to}  \begin{cases}
1 & \text{if } T_i(j) \halts \\
0 & \text{if } T_i(j) \nohalt
\end{cases} 
$$
would be a decider for $K_0$, while in case 2,
$$
\langle i,j \rangle \stackrel{M}{\to} \text{index of }T^{i,j}_k \stackrel{D_{I}}{\to}  \begin{cases}
0 & \text{if } T_i(j) \halts \\
1 & \text{if } T_i(j) \nohalt
\end{cases} 
$$
would be.
\end{proof}

\pscom{All notions for Problem Sheet \#2 have been covered. See \S\ref{sec:ps2}.}

\newpage
\section{Algorithmic Complexity}

We have seen that some functions or numbers cannot be computed. 
But even among computable objects, not all are created equal: some are structured, others seem random. Can we measure this? Algorithmic complexity provides a framework for quantifying the information content of individual strings, not through probabilities, but via their shortest descriptions on a universal machine.


\subsection{Motivation}

There are contexts in which it is unsuitable to look for an underlying probability distribution when communicating a string of symbols.

\begin{example}[Communication of $\pi$]

Suppose Alice wants to communicate the following message to Bob 
$$ \underbrace{31415926538...9}_{10^{10}}~~~~\text{(first decimal digits of $\pi$)} \,.$$ 

Applying Shannon's theory to that case might be frustrating. In practice, the probability distribution is not given; it has to be estimated from the frequencies of the symbols.
 In the case of $\pi$, the estimation would yield a uniform and independent distribution on the alphabet $\{ 0, 1, 2, \ldots 9 \}$.
 In this case, 4 bits per digit are required, and the message is communicated by $4 \cdot 10^{10}$ bits. 
 
Considering groupings of digits doesn’t help much. For instance, Alice could group the digits in blocks of~$3$ and reserve a~$10$-bit codeword for each of the~$1000$ possible sequences of~$3$ digits.
This would allow Alice to communicate the message in~$10 \cdot \frac{10^{10}}{3}$ bits. 
Coding the message with larger blocks would still require more than~$\log_2 10 \cdot 10^{10}$ bits of communication.

The generic idea behind AIT is to assume that the string has been generated by an underlying \emph{computation}; not by an underlying \emph{probabilistic process}. 
In the context of communication, the sender (Alice) and the receiver (Bob) are equipped with universal computers, and Alice may communicate
$$
\texttt{`The first $10^{10}$ digits of~$4 \sum_{n=0}^\infty \frac{(-1)^n}{2n + 1}$,'} 
$$
which Bob runs on his computer to decode.
\end{example}

\begin{example}[Randomness of individual objects]
Is this string random
$$
010101010101010101010101010101010101010101010101010101010101010101\ldots 01 ?
$$
Is this one random
$$
001010110101111010101000011011110101011111001010101111100010011000\ldots 11 ?
$$
\end{example}
We would like to say that the first one is not random and that the second one at least appears to be random. 

How can this be formalized probabilistically?
It is possible that both strings arise from the same probabilistic Bernoulli process of $p=1/2$, in which case both of them have the same probability.

If viewed as arising from a computation instead of a probabilistic process, the patterns displayed by the first string amount to a short description of it. (e.g. print ``01'' $10^6$ times). On the other hand, if the second string really has no patterns, then there is no short description for it. The shortest computable description for it might be:
\begin{center}
\texttt{Print:}
``001010110101111010101000011011110101011111001010101111
\ldots 11\text{''}
\end{center}

Both examples motivate us to look at descriptions of minimal lengths to be run on a universal computer.

\subsection{Plain Complexity}
\lecture{AIT-12}{https://youtu.be/Rxu690XLgoM}
\vspace{5pt}

Algorithmic complexity has two main versions: plain and prefix complexity.\\

 \begin{minipage}{0.55\textwidth}
 \hspace{-10pt}
\qquad \qquad\textbf{Plain complexity $C$}
 \begin{itemize}[itemsep=0.9pt, topsep=5pt]
 \item[--] At first the most intuitive
 \item[--] Originally discovered in the mid-60s by Solomonoff, Kolmogorov and Chaitin. 
 \item[--] Problematic for the development of the theory
 \end{itemize} 
 \end{minipage}
 \begin{minipage}{0.40\textwidth}
  \vspace{-15pt}
\quad \quad\textbf{Prefix complexity $K$}
 \begin{itemize}[itemsep=0.9pt, topsep=5pt]
 \item[--] Initially less intuitive
 \item[--] Discovered in the mid-70s by Levin and Chaitin. 
 \item[--] Solves many problems
 \end{itemize} 
 \end{minipage}

 \vspace{10pt}
The key idea: algorithmic complexity of a bit string $x$ is the length of its shortest description (on a fixed description method).
To be useful, there ought to be a computation (and so a Turing machine) that can generate $x$ from its description:

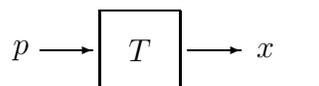
\begin{center} 
\begin{picture}(120,35)(-45,0)
\put(0,0){\framebox(30,30){$T$}} 
\put(-23,15){\vector(1,0){20}}
\put(-40,5){\makebox(20,20){$p$}} 
\put(33,15){\vector(1,0){20}}
\put(50,5){\makebox(25,20){$x$}} 
\end{picture} \,.
\end{center}

\begin{definition}
An input of a Turing machine shall also be called a \emph{program}.
\end{definition}

We can view every Turing machine $T$ as a description method.
If $T$ admits a program $p$ for which $T(p) = x$, then $p$ is a description for $x$ (under the description method given by $T$).

We can thus define the \emph{complexity of a string $x$, with respect to a description method given by} $T$ as
$$
C_T(x) = \min_{p} \{ \len p \st  T(p) = x\} \qquad \text{and} \qquad C_T(x) = \infty \ \text{if no $p$ exists.}
$$

Note that for each Turing machine $T$, we have a complexity measure
$$
x \mapsto C_T(x) \,.
$$
Since what we are after is the shortest descriptions, we have good news:
one complexity measure is ``smaller'' than all the others.




Let $$\langle i,j \rangle \deq \bar i j = 1^i0ij \,.$$ Let $U$ be the universal Turing machine such that 
$
U(\langle i, j \rangle) = T_i(j) \,.
$

\begin{definition}
The \emph{plain algorithmic complexity} $C_U(x)$ of a string $x \in \{0,1\}^*$ is defined as
$$
C_U(x) = \min_{p} \{ \len p \st  U(p) = x\} \,.
$$
\end{definition}

\subsection{The Invariance Theorem}
\begin{theorem}[Additive optimality of $U$; The invariance theorem; Universality of algorithmic complexity]
For any Turing machine $T$, there exists a constant $b_T$ such that,
\begin{equation} \tag{*}
C_U(x) \leq C_T(x) + b_T ~~~~\forall x\,. 
\end{equation}
\end{theorem}
\begin{proof}
If $C_T(x) \leq \infty$, it means that the machine $T$ can compute $x$ by some program $q$, where $\len q = C_T(x)$. \\
But $T = T_i$ for some $i$.
Therefore, $U$ can also compute $x$, via the program $p = \langle i, q \rangle$. This program has length
$$
\len{ 1^i0iq } = \len q + \underbrace{2 i + 1}_{b_T} = C_T(x) + b_T \,.
$$
\end{proof}
A universal Turing machine with the property (*) is \emph{additively optimal}.
Note that not all universal Turing machines are additively optimal. 
For instance, a universal Turing machine $W$ may expect input of the form $1^i01^q$ to simulate $T_i(q)$ (and halt with empty output on any program not of the right form).

The invariance theorem can be applied by setting $T=V$, where $V$ is another additively optimal universal Turing machine. We have $\exists b_V$, $\exists b_U$ such that
$$C_U(x) \leq C_V(x) + b_V  ~~~~\text{and}~~~~
C_V(x) \leq C_U(x) + b_U ~~~~\forall x\,.  $$
Therefore, there exist a constant $b_{UV}$ such that
$$
|C_V(x) - C_U(x)| \leq b_{UV}~~~~\forall x\,.
$$

\textbf{Notation:}
We denote~$O(f(x))$ to be a quantity whose absolute value does not exceed~$f(x)$ by more than a fixed multiplicative factor. More precisely,
\mbox{$g(x) = O(f(x))$} if there are constants~$c$, $x_0$ such that \mbox{$|g(x)| \leq c f(x)$} for all~$x \geq x_0$.
This permits us to write
$$
C_V(x) = C_U(x) + O(1) \,.
$$

We thus see that the choice of reference universal and additively optimal Turing machine only matters by an additive constant (independent of $x$). 
We consider this irrelevant, as it is irrelevant for strings large enough.
This is the invariance in the choice of the reference machine.
The invariance theorem motivates dropping the~$U$ in~$C_U$. We simply write~$C(x)$.

The invariance theorem is particularly useful when the Church--Turing is invoked.
What is the length~$C_{\text{LISP}} (x)$ of the minimal LISP program for a string~$x$? 
What is the length~$C_{\text{PYTHON}} (x)$ of the minimal PYTHON program for a string~$x$? 
The invariance theorem states that there exists a universal constant $b_{\text{LISP} \,;\, \text{PYTHON}}$ that bounds the difference between the measures $C_{\text{LISP}} (\cdot)$ and $C_{\text{PYTHON}} (\cdot)$, i.e.,
$$
|C_{\text{LISP}} (x) - C_{\text{PYTHON}} (x)| \leq b_{\text{LISP}\,;\, \text{PYTHON}} \qquad \forall x \,.
$$
The constant may be large, but for strings large enough, it is irrelevant.

And by the Church--Turing--Deutsch principle, if any other physical system is used to define a complexity measure on strings, it will not yield a more minimal measure than~$C_U$.

\subsection{Basic Properties}
\lecture{AIT-13}{https://youtu.be/fKcz5x5EmgE}

\begin{definition}
Let $x$ be a finite string. Its \emph{shortest program} is denoted~$x^*$. In other words,~$x^*$ is the witness of~$C(x)$. If there is more than one shortest program, $x^*$ is the one with the least running time. 
\end{definition}

\begin{proposition}[upper bound]
For all strings $x$,
$$C(x) \leq \len x + O(1)\,.$$
\end{proposition}

\begin{proof}
\nit{Invoking Turing machines.}\\
There is a Turing machine $Id$, which computes the identity function. $U$ can simulate $Id $ on input $x$. The program for $U$ would be $p = \langle i, x \rangle = \bar i x$, where $T_i = Id$.
$$
|p| = |\bar i| + |x| = |x| + O(1) \,.
$$
This is \emph{one} program for computing $x$ from $U$: the minimal program therefore has length $C(x) \leq |p|$.

\nit{Invoking high-level programming language.}

\texttt{Print ``$x$''} is a program for $x$. It has length $|x| + O(1)$.
\end{proof}

Recalling our identification of strings with numbers, Eq~\eqref{eq:bijection}, and Prop.~\ref{prop:length}, \mbox{$|x| = \lfloor \log x \rfloor$}, so
$$C(x) \leq \log x + O(1)\,.$$

\begin{definition}
The \emph{complexity} $C(x,y)$ of a pair of strings is defined with respect to a fixed encoding $\langle \cdot, \cdot \rangle : \{0,1\}^* \times \{0,1\}^* \to \{0,1\}^* $ as
$$
C(x,y) = \min_{p} \{ \len p \st  U(p) = \langle x, y \rangle\} \,.
$$
\end{definition}

\begin{example}
$C(x,x) = C(x) + O(1)$.
This means two things:
$$1.~~C(x,x) \leq C(x) + O(1) \qquad \text{and} \qquad
2.~~C(x) \leq C(x,x) + O(1) \,.$$
\noindent 1. One program to compute $\langle x, x \rangle$ can be obtained from $x^*$ by executing $x^*$ to get $x$, and computing $\langle x, x \rangle$.\\
\noindent 2. One program to compute $x$ can be obtained from $\langle x, x \rangle^*$ by executing $\langle x, x \rangle^*$ and extract $x$ from the inverse of $\langle \cdot , \cdot \rangle$.
\end{example}

The computation of a string $x$ may be eased by some auxiliary information $y$.
\begin{definition}
The \emph{conditional complexity} of the string $x$ given the string $y$ is
$$
C(x \given y) \deq \min_{p} \{ \len p \st  U(p,y) = x\}  \,,
$$
where the universal Turing machine $U(\cdot, \cdot)$ has been promoted to two tapes\footnote{The two-tape Turing machine can be simulated by a single-tape Turing machine; so two-tape machines compute the same class of functions. I invoke two-tape machines because they correspond to the realistic model of keeping auxiliary information and programs on different registers.}.
The second entry refers to the tape of \emph{auxiliary information}.
\end{definition}
We update the unconditional definition to
$
C(x) \deq C(x \given \varepsilon) \,.
$

\begin{example}
$C(x \given y) \leq C(x) + O(1)$,
because one way to compute $x$ with $y$ as auxiliary information is to ignore $y$, and produce $x$ from scratch with $x^*$.
\end{example}

\subsection{Chain Rule for Plain Complexity}
\begin{proposition}[Chain rule, easy side]\label{creasy}
For all strings $x$ and $y$,
$$C(x, y) \leq C(x) + C(y \given x) + O(\log n)\,,$$
where $n = \max \{\len x, \len y\}$.
\end{proposition}
\begin{proof}
Let $q$ be the witness of $C(y \given x)$, i.e., the program of minimal length such that $U(q,x)=y$. 
The aim is to construct a program for $\langle x, y \rangle$ from $x^*$ and $q$.
If delimiters (\#) were free, a program for $\langle x, y \rangle$ would be 
$$
r\#x^*\#q \,,
$$
where $r$ contains the following instructions:
\begin{enumerate}[itemsep=0.9pt, topsep=5pt]
\item Between the first and second delimiter, a program ($x^*$) is given to compute the first string ($x$), and this string should be stored on the auxiliary tape.
\item After the second delimiter, a program ($q$) is given to compute the second string ($y$) with the help of $x$ on the auxiliary tape.
\item Compute the pairing $\langle x, y\rangle$.
\end{enumerate}

Note that $r$ is of constant length (independent of $x$ and $y$, so independent of $n$). However, delimiters are not free: what we want is a single program $p$, made only of bits, for which $U(p)= \langle x, y\rangle$.
If $p$ is of the form
$$
p= r'x^*q \,,
$$
the preamble $r'$ needs to do what $r$ does, and moreover, it needs to contain some information to break $x^*q$ into $x^*$ and $q$. It can be the length of $x^*$, ($C(x)$), given in a self-delimiting way (e.g. $\overline{C(x)}$).
Thus
\beas
\len p &=& \len{r'} + \len{x^*} + \len q\\
&=& \len{r'} + C(x) + C(y \given x)\\
&\leq& C(x) + C(y \given x) + 2 \log C(x) + O(1)\\
&\leq& C(x) + C(y \given x) + O(\log n) \,.
\eeas

\end{proof}

\lecture{AIT-14}{https://youtu.be/UsjitZb7_hA}

\begin{proposition}[Chain rule for plain complexity, hard side]\label{crhard}
For all strings $x$ and~$y$,
\beas
C(x) + C(y \given x) &\leq& C(x, y) + O(\log (C(x,y)))\\
&\leq& C(x, y) + O(\log n)
\eeas
where again $n = \max \{\len x, \len y\}$.
\end{proposition}

\begin{proof}
Let $m \equiv C(x,y)$.
Define
\bes
A = \{ \langle x',y' \rangle \st C(x',y') \leq m \} \qquad \text{and} \qquad 
A_x = \{ y' \st C(x,y') \leq m \} \,,
\ees
which contain $\langle x,y \rangle$ and $y$, respectively. 
A program for $y$ given $x$ is to computably enumerate $A_x$ and to give its enumeration number~$i^y$.
$A_x$ can be enumerated if $x$ and $m$ are known, by running all programs of length $m$ and if a program halts with an output of the form $\langle x, y'\rangle$, listing $y'$.

 Let $l \equiv \lceil \log |A_x| \rceil $ be an upper bound on the number of bits of~$i^y$. 
 Hence, the program considered requires $O(\log m)$ bits to compute $m$, and $l$ bits to give the enumeration number of $y$. 
 So 
\be \label{eq:dif1}
C(y \given x) \leq l + O(\log m) \,.
\ee

\noindent Define now
\bes
B = \{ x' \st \log |A_{x'}| > l-1 \} \,,
\ees
which contains $x$. If $m$ and $l$ are given (with an $O(\log m)$ advice, because $l$ is a number smaller than $m+1$, because $|A_x|$ is smaller than the number of programs of length $\leq m$), $B$ can be enumerated by enumerating $A$ (thanks to $m$), and when  for a given~$x'$ the subset $A_{x'}$ contains more than $2^{l-1}$ elements, $x'$ is added to $B$. A possible program for $x$ is thus given by the enumeration number of $x$ in $B$. 
 Note that
\bes
|A| = \sum_{x'} |A_{x'}| \geq \sum_{x'\in B} |A_{x'}| \geq \sum_{x'\in B} 2^{l-1} = |B| 2^{l-1}\,.
\ees
Since $|A| < 2^{m+1}$,
\bes
\log |B| \leq m - l +2  \,.
\ees
Therefore
\be \label{eq:dif2}
C(x) \leq m-l +  O(\log m)  \,.
\ee
\noindent Summing \eqref{eq:dif1} and \eqref{eq:dif2} together yields what is to be shown.
\end{proof}

From Propositions \ref{creasy} and \ref{crhard}, this version of the chain rule follows 
\begin{theorem}[Chain rule]
For all strings $x$ and $y$,
\beas
C(x, y) = C(x) + C(y \given x) + O(\log (C(x,y))) \,.\\
\eeas
\end{theorem}
Again, the error term could be set to be $O(\log n)$, where $n = \max \{\len x, \len y\}$.

\subsection{Incompressibility}
It is easy to come up with highly compressible strings.
For instance, define
$$
f(k) = 2^{2^{2^{\dots^{2}}}} \qquad k\text{ times}
$$

 For each~$k$, the binary representation of~$f(k)$ has complexity at most~$C(k) + c \leq \log k + c$ for some constant~$c$ independent of~$k$. Yet it has a huge length. 

What about incompressible strings?
\begin{definition}\label{def:incompressible}
For each $c$, a string~$x$ is $c$\emph{-incompressible} (given $y$) if $C(x)\geq \len x -c$ (if $C(x \given y)\geq \len x -c$).
\end{definition}

\begin{proposition}\label{prop:incompressible}
There are at least $2^n - 2^{n-c} + 1$ strings of length $n$ that are $c$-incompressible.
\end{proposition}
\begin{proof}
The number of strings of length $n$ is $2^n$.
The number of programs of length strictly less than $n - c$ is
$$
\sum_{i=0}^{n-c-1}2^i = 2^{n-c}-1
$$
In the most conservative scenario, each of those short programs computes a different string of length $n$. Indeed, generically, there are some of those short programs that don't produce any output (they don't halt). Also, there are some of those short programs that compute strings of a length other than $n$ bits. And finally, there are some of those short programs that compute the same string.

Thus, there are at least $2^n 2^{n-c} + 1$ strings that are \emph{not} computed by such a short program, namely, that are $c$-incompressible.
\end{proof}
\begin{corollary}
There is at least $1$ string that is $0$-incompressible. Not so bad. But it goes fast: more than 3/4 of the strings are $2$-incompressible.
\end{corollary}
\begin{remark}
The definition of incompressibility can be made to accommodate auxiliary information: a string~$x$ is $c$\emph{-incompressible} \emph{given $y$} if $C(x \given y)\geq \len x -c$. And Proposition~\ref{prop:incompressible} carries through in the presence of auxiliary informaiton~$y$.
\end{remark}

For our purposes, we can identify incompressible strings with algorithmically random strings.

\subsection{The Problems of Plain complexity}\label{sec:problemswithplain}

\begin{enumerate}[itemsep=0.9pt, topsep=5pt]
\item The chain rule holds only up to $O(\log C(x,y))$.
\item Let $n$ be an incompressible number (string) of length $m$. Consider
$$
\underbrace{00\ldots0}_{n \text{ times}} \equiv 0^n \,.
$$ 
$$
C(0^n) = C(n) + O(1) = m + O(1) \,.
$$
As in the case of $0^n$, there can be algorithmic complexity arising from the length of the string, even when the string’s contents are highly regular.
%
Consider now $x\in \{0,1\}^n$, incompressible given $n$, i.e., $C(x \given n) \geq n$. 
The fact that not only $C(x) \geq n$ but also $C(x \given n) \geq n$ means that the algorithmic information present in the length of $x$, which is by default carried by $x$, is independent of the algorithmic information present within $x$.
Intuitively, we would like that
$$
C(x) = n + m + O(1),
$$ 
where $m$ bits of algorithmic complexity are embodied in the size of $x$ and
$n$ bits of algorithmic complexity come from the bits within $x$.
But we have shown that $C(x) \leq n + O(1)$.


%


\item Solomonoff had a brilliant idea for a measure over all strings. 
Instead of a monkey in front of a typewriter, consider a monkey in front of the computer, and define:
$$
x \mapsto P(x) = \sum_{p \st U(p) = x} 2^{-|p|} \,.
$$ 
This measure amounts to the probability that a program for which the bits are picked at random would yield the string $x$. 

Solomonoff's problem was that the measure does not converge. Indeed there exists a $T_i$ such that on all strings $y$, $T_i(y)=x$. Thus, each $\bar i y$ is a valid program for $x$, as $U(\bar i y)=x$ for all $y \in \{0,1\}^*$.
\beas
P(x) &\geq& \sum_{y \in \{0,1\}^*} 2^{-\len{\bar i y}} \\
&=& 2^{-\len{\bar i}} \sum_{y \in \{0,1\}^*} 2^{-\len y} \\
&=& 2^{-\len{\bar i}} \sum_{n = 0}^{\infty} \sum_{y \in \{0,1\}^n}^{\infty} 2^{-n} \\
&=& 2^{-\len{\bar i}} \sum_{n = 0}^{\infty}  2^{-n} 2^n \\
&=& \infty
\,.
\eeas

\clearpage
\lecture{AIT-15}{https://youtu.be/-K57aqoGNug}
\item Recalling coding theory,
\beas
SP \colon \{0,1\}^* &\to& \{0,1\}^*\\
x &\mapsto& x^*
\eeas  
is a natural and powerful code. Is it a prefix code? No. In fact, considering $y \neq x$, nothing prevents $y^*$ from being a prefix of $x^*$ because a priori, nothing prevents a program from being the prefix of another program. 

\end{enumerate}

Adopting prefix instead of plain complexity solves all of the above problems.

\newpage
\section{Prefix Complexity}

In prefix complexity, it is demanded that a program must be \emph{self-delimited}: its total length in bits must be determined by the program itself. 
For Turing machines, it means that the head scanning the program is not allowed to overshoot. 

This way, the alphabet on which we measure program size really is binary. In the plain complexity framework, the head can hit and recognize the first blank after the program, which effectively introduces a third symbol into the alphabet used to describe programs. 

Real programming languages are self-delimiting because they provide instructions for beginning and ending programs.
This allows a program to be well-formed by the concatenation or nesting of sub-programs.

\subsection{Self-delimiting Turing Machines}\label{sec:sdtm}

A \emph{self-delimiting Turing machine} has three tapes: 
\begin{enumerate}[itemsep=0.9pt, topsep=5pt]
\item The \emph{input} or \emph{program} tape is one-way (the head can only move to the right), read-only and initially contains the program.
\item The \emph{work} tape moves two-ways, it reads and writes and initially contains a string~$y$ as auxiliary information.
\item The \emph{output} tape is one-way, write-only and initially empty.
\end{enumerate}
If, after finitely many steps, $T$ halts with 
\begin{itemize}[itemsep=0.9pt, topsep=5pt]
\item[--] a string~$p$ that has been so far scanned on the input tape, namely, upon halting, the head of the input tape is on the last bit of $p$, but no further;
\item[--] possibly, a string $y$ that was initially given on the work tape;
\item[--] a string $x$ on the output tape;
\end{itemize}
then we say that the \emph{computation is successful}, and $T(p, y) = x$, (in particular, $T(p, y) \halts$). 
Otherwise, the \emph{computation is a failure}, and $T(p, y) \nohalt$.

Self-delimiting Turing machines can be effectively enumerated, $T_1$, $T_2$, \dots, and there exists $U$ s.t.
$$
U(\bar i q) = T_i(q) \qquad \forall i, q \,.
$$ 
More generally, with auxiliary information, 
$$
U(\bar i q, y) = T_i(q, y) \qquad \forall i, q , y \,.
$$ 
The first argument refers to the program tape content; the second to the work (auxiliary) tape.
\begin{remark}
For all self-delimiting Turing machines~$T$ and for all auxiliary information~$y$, the set of programs that lead to a halting computation $\{p \st T(p,y) \halts \}$ is a prefix-free set. 
Indeed, if $T(q,y) \halts$, it means that the last bit of $q$ has been read (and no more). Therefore, no program of the form $qs$, for some non-empty string $s$, can also lead to a successful computation\footnote{If $qs$ happens to be on the input tape, the computation would carry through, but $q$ would be recognized as being the program of the computation. Therefore, we adopt the convention that $T(qs,y)$ is undefined and also denote it as $\nohalt$.}.

In particular, the set $\{p \st U(p) \halts \}$ is a prefix-free set.
See Fig.~\ref{fig:tree2}.
\begin{figure}[h]
\centering
\begin{tikzpicture}
\node[circle, draw, fill=black, inner sep=0.5pt, name=root] at (0,0) {};

\node[circle, draw, fill=black, inner sep=0.5pt, name=n0] at (1,1) {};
\node[circle, draw, fill=black, inner sep=0.5pt, name=n1] at (1,-1) {};
\draw (root) -- node[midway, above] {\scriptsize 0} (n0);
\draw (root) -- node[midway, above] {\scriptsize 1} (n1);

\node[circle, draw, fill=black, inner sep=0.5pt, name=n00] at (2,1.5) {};
\node[circle, draw, fill=black, inner sep=0.5pt, name=n01] at (2,0.5) {};
\node[right=-1pt of n01] {$\dots$};
\draw (n0) -- node[midway, above] {\scriptsize 0} (n00);
\draw (n0) -- node[midway, above] {\scriptsize 1} (n01);

\node[circle, draw, fill=black, inner sep=1.5pt, name=n10] at (2,-0.5) {};
\node[right=6pt of n10] {$U(10) = \varepsilon$};
\node[circle, draw, fill=black, inner sep=0.5pt, name=n11] at (2,-1.5) {};
\draw (n1) -- node[midway, above] {\scriptsize 0} (n10);
\draw (n1) -- node[midway, above] {\scriptsize 1} (n11);

\node[circle, draw, fill=black, inner sep=0.5pt, name=n000] at (3,1.75) {};
\node[right=-1pt of n000] {$\dots$};
\node[circle, draw, fill=black, inner sep=1.5pt, name=n001] at (3,1.25) {};
\node[right=6pt of n001] {$U(001) = 11$};
\draw (n00) -- node[midway, yshift=3pt] {\tiny 0} (n000);
\draw (n00) -- node[midway, yshift=3pt] {\tiny 1} (n001);

\node[name=n110] at (3,-1.25) {};
\node[right=-10pt of n110] {$\circlearrowleft$};
\node[right=6pt of n110] {$U(110) \nohalt 
$};
\node[circle, draw, fill=black, inner sep=0.5pt, name=n111] at (3,-1.75) {};
\node[right=-1pt of n111] {$\dots$};
\draw (n11) -- node[midway, yshift=3pt] {\tiny 0} (n110);
\draw (n11) -- node[midway, yshift=3pt] {\tiny 1} (n111);
\end{tikzpicture}
\caption{A possible prefix tree representing the set $\{p \st U(p) \halts \}$. On input~$001$, the self-delimiting universal Turing machine~$U$ halts and outputs~$11$, and on input~$10$, it outputs~$\varepsilon$. Thus, all programs of the form $001s$ or $11s$ for non-empty $s$ do not lead to a successful computation because their last bit will never be scanned. We denote~$U(001s) \nohalt$ and~$U(10s) \nohalt$, and remove the descendant branches from the tree.
Moreover, the program~$110$ loops. Therefore, $U(110)  \nohalt$, but also $U(001s)  \nohalt$ for all nonempty $s$.
 }
\label{fig:tree2}
\end{figure}

\end{remark}
%
\begin{definition}
The \emph{conditional prefix algorithmic complexity} $K(x \given y)$ of a string $x \in \{0,1\}^*$ is defined as
$$
K(x\given y) = \min_{p} \{ \len p \st  U(p, y) = x\} \,.
$$
\end{definition}
As in plain complexity, $U$ is additively optimal, and the invariance theorem holds: $\forall$ self-delimiting Turing machines $T$ and $\forall$ strings $x, y$
$$
K(x \given y) \deq K_U(x\given y) \leq K_T(x\given y) + O(1) \,.
$$
\begin{proof}
If $K_T(x \given y) \leq \infty$, it means that $T(q, y)=x$ for some (self-delimiting) program~$q$, where~$\len q = K_T(x)$.
But~$T = T_i$ for some~$i$.
Therefore,~$U$ can also compute~$x$, via the program~$p = \bar i q$. This program has a length
$$
\len{ 1^i0iq } = \len q + \underbrace{2 i + 1}_{b_T} = K_T(x) + b_T \,.
$$
\end{proof}

The \emph{unconditional prefix complexity} of $x$ is $K(x) \deq K(x \given \epsilon)$, $x^*$ is the shortest program, and the \emph{prefix complexity of a pair} is $K(x,y) \deq K(\langle x, y \rangle)$.

\subsection{Solving the Problems of Plain Complexity}
Let us address the problems of \S \ref{sec:problemswithplain} in reverse order. We now consider the notions in the context of prefix complexity.
\begin{enumerate}[itemsep=0.9pt, topsep=5pt]
\item[4.] 
The code
\beas
SP \colon \{0,1\}^* &\to& \{0,1\}^*\\
x &\mapsto& x^*
\eeas
is a prefix code.
\end{enumerate}

\lecture{AIT-16}{https://youtu.be/HPtAqpTo3Xs}\\
\nit{Comment on the lecture: Unfortunately,  the audio was not recorded.}

\begin{enumerate}[itemsep=0.9pt, topsep=5pt]
\item[3.] Solomonoff's problem is solved:
$$
x \mapsto P(x) = \sum_{p \st U(p) = x} 2^{-|p|} \,
$$ 
now converges for all $x$. Indeed, the set $\{p \st U(p) = x\}$ 
is prefix-free because it is a subset of the prefix-free domain of halting programs
 $\{p \st U(p) \halts\}$. Therefore, by Kraft's inequality not only does $P(x) \leq 1$, but also $\sum_x P(x) = \sum_{\{p \st U(p) \halts\}} 2^{-|p|} \leq 1$.

\item[2.] Let $x\in \{0,1\}^n$. I mentioned that intuitively, if both $x$ and $n$ are algorithmically random, we might want
$$
K(x) = n + m + O(1) \,,
$$
where $m$ measures to the complexity of $n$ itself, while $n$ measures the complexity of the bits ``within'' $x$.
However, with plain complexity, the \texttt{print ``$x$''} program entailed $C(x) \leq n + O(1)$. Yet, as we had seen it, the print program was not self-delimiting.
\end{enumerate}
\begin{proposition}\label{prop:uppercrefix}
For all $x\in \{0,1\}^n$, 
$$
K(x) \leq n + K(n) + O(1) \,.
$$
\end{proposition}
\begin{proof}
Consider the following self-delimiting program for $U$:
$$
rn^*x \,,
$$
where $r$ is a constant-sized (and self-delimited) program which gives the following instructions. 
\begin{enumerate}[itemsep=0.9pt, topsep=5pt]
\item Execute $n^*$ to find $n$. (Note that there is no need for delimiting information to tell $n^*$ and $x$ apart, because $n^*$ is self-delimiting.)
\item Output the next $n$ bits on the tape ($x$).
\end{enumerate}
\end{proof}
\begin{corollary}
Let $x \in\{0,1\}^*$. Via iterative appeals of Prop.~\ref{prop:uppercrefix}, we find that
\beas
K(x) &\leq& \len x + K(\len x) + O(1)\\
&\leq& \len x + \len{\len x} + K(\len{\len x}) + O(1)\\
&\leq& \len x + \len{\len x} + \len{\len{\len x}} + K(\len{\len{\len x}}) + O(1)\\
&\leq& \log x + \log{\log x} + \log{\log \log x} + O(\log \log \log \log x) \,,
\eeas
where $||\cdot||$ denotes the length of the length.
\end{corollary}
This gives an upper bound for~$K(x)$ that is nearly tight. Indeed, as is shown in Problem Sheet~\#3 (\S\ref{sec:ps3}), for infinitely many values of~$x$, $K(x)$ exceeds~$\log x + \log{\log x} + \log{\log \log x}$.
See Fig.~\ref{fig:kgraph}.

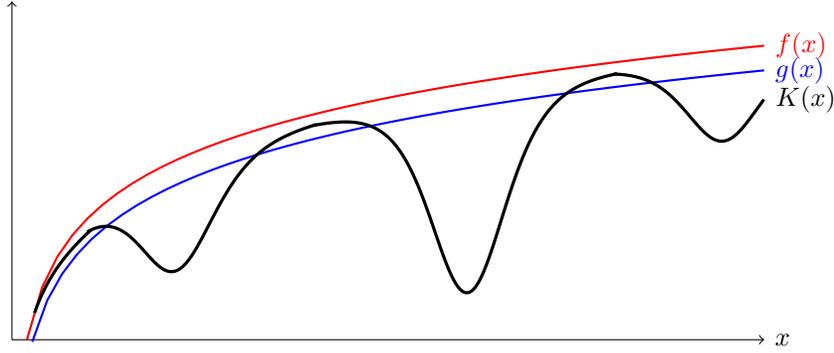
\begin{figure}
\centering
\begin{tikzpicture}
    \draw[->] (0,0) -- (10,0) node[right] {\footnotesize $x$};
    \draw[->] (0,0) -- (0,4.5) node[above] {\footnotesize $$};

    \draw[red, thick, domain=0.2:10, samples=50] 
        plot (\x, {ln(5*\x)}) 
        node[right] {\footnotesize $f(x)$};

    \draw[blue, thick, domain=0.27:10, samples=50] 
        plot (\x, {ln(3.6*\x)}) 
        node[right] {\footnotesize $g(x)$};

     \draw[black, very thick, domain=0.3:1.05, samples=350] 
        plot (\x, {ln(5*\x)*(1- 0.2*(1/((1 + exp(-2*(\x - 0.5)))*(1 + exp(2*(\x - 0.7)))))-0.07});
    \draw[black, very thick, domain=1:4.05, samples=350] 
        plot (\x, {ln(5*\x)*(1- 2*(1/((1 + exp(-3*(\x - 2.1)))*(1 + exp(3*(\x - 2.2)))))-0.04});
     \draw[black, very thick, domain=4:8.05, samples=350] 
        plot (\x, {ln(5*\x)*(1- 2.7*(1/((1 + exp(-3*(\x - 6)))*(1 + exp(3*(\x - 6.1))))))-0.126});
        \draw[black, very thick, domain=8:10, samples=350] 
        plot (\x, {ln(5*\x)*(1- 1*(1/((1 + exp(-3*(\x - 9.4)))*(1 + exp(3*(\x - 9.5))))))-0.1})
        node[right] {\footnotesize $K(x)$};

\end{tikzpicture}
\caption{A made-up graph for $K(x)$, with its upper bound given by $f(x)=\log x + \log{\log x} + \log{\log \log x} + O(\log \log \log \log x)$, which is very close to being tight: infinitely often, $K(x)$ exceeds $g(x) = \log x + \log{\log x} + \log{\log \log x} $.}
\label{fig:kgraph}
\end{figure}

\begin{enumerate}[itemsep=0.9pt, topsep=5pt]
\item[1.] The chain rule had a $O(\log C(x,y))$ error term. For prefix complexity, the error is~$O(1)$.
\end{enumerate}

\begin{theorem}[Chain rule for prefix complexity]\label{thm:cro1}
For all strings $x$ and $y$,
$$K(x,y) = K(x) + K(y \given x^*) + O(1) \,.$$
\end{theorem}
The proof of the chain rule follows upon proving the easy side, Prop.~\ref{prop:creasy}, and the hard side, Prop.~\ref{prop:crhard} in \S\ref{sec:lss}. 

\begin{proposition}[Chain rule for prefix complexity, easy side] \label{prop:creasy}
For all $x$ and $y$,
$$
K(x,y) \leq K(x) + K(y \given x^*) + O(1) \,.
$$
\end{proposition}
\begin{proof}
\pscom{In Problem Sheet \#3. See \ref{sec:ps3}.}
\end{proof}

\subsection{Information Equivalence between $x^*$ and $(x, K(x))$}
\lecture{AIT-17}{https://youtu.be/T3jtJDtrUK8}

\begin{proposition}\label{prop:infoeq}
There is the same algorithmic information in $x^*$ as there is in $(x, K(x))$, i.e.
$$K(\langle x, K(x)\rangle \given x^*) =O(1) \qquad \text{and} \qquad  K(x^* \given\langle x, K(x)\rangle  ) = O(1) \,.$$
\end{proposition}
\noindent Here is a convenient notation for the above relations:
\begin{center} 
\begin{picture}(120,40)(-45,0)
\put(0,0){\framebox(30,30){$O(1)$}} 
\put(-23,15){\vector(1,0){20}}
\put(-43,5){\makebox(20,20){$x^*$}} 
\put(33,23){\vector(1,0){20}}
\put(33,7){\vector(1,0){20}}
\put(53,13){\makebox(20,20){$x$}} 
\put(58,-3){\makebox(20,20){$K(x)$}} 
\end{picture}\qquad \text{and} \qquad
\begin{picture}(120,40)(-45,0)
\put(0,0){\framebox(30,30){$O(1)$}} 
\put(-23,23){\vector(1,0){20}}
\put(-23,7){\vector(1,0){20}}
\put(-43,13){\makebox(20,20){$x$}} 
\put(-50,-3){\makebox(20,20){$K(x)$}} 
\put(33,15){\vector(1,0){20}}
\put(55,5){\makebox(20,20){$x^*$}} 
\end{picture}~~.
\end{center}

\begin{proof}
First, we want to show that there exists a constant-size program, which takes $x^*$ as auxiliary information and outputs $\langle x, K(x) \rangle$. This can be done with the following instructions:
\begin{enumerate}[itemsep=0.9pt, topsep=5pt]
\item Measure the length of $x^*$ (this gives $K(x)$);
\item Execute $x^*$ to get $x$;
\item Compute $\langle x, K(x) \rangle$.
\end{enumerate}

Second, to compute $x^*$ from $(x, K(x))$, for $i=1,2, \ldots$, run all programs of length $K(x)$ for $i$ steps. Eventually some $p$ halts with $U(p) = x$, this $p$ is $x^*$.
\end{proof}

Proposition~\ref{prop:infoeq} entails that the chain rule can be expressed as
$$K(x,y) = K(x) + K(y \given x, K(x)) + O(1) \,.$$

Prefix complexity may be slightly larger than plain complexity, due to the cost of self-delimiting structure, but not by more than~$O(\log n)$ for strings of length~$n$.
\begin{proposition}
For all $x \in \{0,1\}^n$, $K(x) = C(x) + O(\log n)$.
\end{proposition}
\begin{proof}
\pscom{See Problem Sheet \#3, \S\ref{sec:ps3}.}
\end{proof}

%
%
%
%

\subsection{The Uncomputability of $K(x)$}

\begin{theorem}
The function $x \mapsto K(x)$ is uncomputable.
\end{theorem}
\begin{proof}
Assume that there is a Turing machine $M$ which, given $x$, computes~$K(x)$. Let~$B$ be a constant to determine.
Consider the following program. 

\begin{quote}\texttt{Simulate $M$ to compute $K(x)$ for $x = 1, 2 , 3, \ldots$ until some $z_B$ is found for which $K(z_B) > B$. Output that $z_B$.} \end{quote}
What is the size of the above program? 
\begin{itemize}[itemsep=0.9pt, topsep=5pt]
\item[--] The instructions to simulate~$M$ are of constant size~$c_1$;
\item[--] The hard coding of the constant~$B$ can be given in a self-delimiting way by~$2 \log B +1$ bits;
\item[--] The extra instructions to assemble the program is of constant size $c_2$.
\end{itemize}

Therefore the size of the program is $c_1 + c_2 + 2 \log B +1$. 
Let $B$ be any number such that 
$$
c_1 + c_2 + 2 \log B +1 < B.
$$

Thus, we have a program of size $\leq B$ that produces $z_B$, which, supposedly, has $K(z_B) > B$. Contradiction.
\end{proof}

This proof is the formalization of the Berry paradox.
In \S\ref{sec:bp}, we defined~$k_{20}$ as ``the smallest number that cannot be described by twenty syllables or less''. This ``definition'' puts a high burden on what is meant by ``described''. In AIT, we look at algorithmic descriptions, namely those that are given to a universal computer.

The proof of the uncomputablility of~$K(x)$ involves the algorithmic-information-theoretic cousin of~$k_{20}$. Indeed,~$z_B = \min \{x \st x \in \{0,1\}^* \, \& \, K(x) > B\}$, or in words, it is 
\ct{``the smallest number that cannot be \emph{algorithmically} described by~$B$ bits or less.''}
The displayed sentence can be put into bits, and if $B$ is large enough, the obtained description would be shorter than $B$ bits. The contradiction is avoided because while the displayed sentence encoded in bits is a description of $z_B$---it uniquely specifies it---by the uncomputability of~$K(x)$, it is not an \emph{algorithmic} description.

\newpage
\section{The Halting Probability $\Omega$}

As we shall see in that section, the halting probability~$\Omega$---also known as Chaitin’s constant---is an algorithmically random number that compactly encodes the halting problem and, with it, many mathematical truths.

Reminder: $U(p) = T_i(j)$, where $p = \bar i j$. We defined Turing's number $\chi$ as 
$$
\chi = \chi_0 \chi_1 \chi_2 \ldots \,,
$$
where
$$ 
\chi_p = \begin{cases}
 1 & \text{if } T_i(j) \halts \\
 0 & \text{if } T_i(j) \nohalt  \,.
 \end{cases}
 = \begin{cases}
 1 & \text{if } U(p) \halts \\
 0 & \text{if } U(p) \nohalt  \,.
 \end{cases}
 $$
 Observe that the first $2^0 + 2^1 + 2^2 + \ldots 2^n = 2^{n+1}-1$ bits of $\chi$ encode the halting problem for all programs of length $ \leq n$.

\begin{definition}
The \emph{halting probability} $\Omega$ (or \emph{Chaitin's constant}) is
$$
\Omega \deq \sum_{p \st U(p) \halts}2^{- \len p} \,.
$$
\end{definition}
It corresponds to the probability that the self-delimiting universal Turing machine halts if it is given a program whose bits are chosen by coin flip. Note that by the construction of self-delimiting machines (see~\S\ref{sec:sdtm}), a random bit can be drawn only when the reading head asks for one more bit, thus keeping the process finite.

\vspace{10pt}
\lecture{AIT-18}{https://youtu.be/iPZXGbPodt4}\\
\nit{Comment on the lecture: 
The solution to Problem 4 of Sheet \#2 (\ref{sec:ps2}) is discussed. 
Moreover, in the first hour, I reviewed a selection of articles. See Appendix~\ref{sec:papers} for the references and abstracts of the presented articles.}

\begin{proposition}\label{prop:compress}
$\Omega$ is a compressed version of $\chi$: $n$ bits of $\Omega$ are enough to compute $2^{n+1}-1$ bits of $\chi$. Schematically,
\begin{center} 
\begin{picture}(120,35)(-45,0)
\put(0,0){\framebox(30,30){$O(1)$}} 
\put(-23,15){\vector(1,0){20}}
\put(-48,5){\makebox(20,20){$\finite{\Omega}{n}$}} 
\put(33,15){\vector(1,0){20}}
\put(65,5){\makebox(25,20){$\finite{\chi}{2^{n+1}-1}$}} 
\end{picture} ~~~~~.
\end{center}
\end{proposition}

\begin{proof}
 For $i=1,2, \ldots$ run all programs of length $\leq i$ for $i$ steps of computation. 
 (By the way, this is called ``running all programs in a \emph{dovetailed fashion}'').
Add~$2^{-|p|}$ to a sum~$S$ (initially set to~$0$) whenever a program~$p$ halts. 
When the first $n$ bits of the sum~$S$ stabilize to the first $n$ bits of $\Omega$, \emph{i.e.}, $\finite{S}{n}=\finite{\Omega}{n}$, (this happens when $S\geq \finite{\Omega}{n}$) then no program of length $\leq n$ will ever halt, since such an additional contribution to the sum would contradict the value of~$\Omega$. 
This process is said to \emph{lower semi-compute}~$\Omega$, since it always returns smaller numbers than~$\Omega$, and they converge to it in the limit of infinite time.
\end{proof}
Thus, $\Omega$ is not computable, but it is lower semi-computable: we can successively approximate it from below, though never exactly determine its bits.

\vspace{5pt}
\lecture{AIT-19}{https://youtu.be/mxliDkzo6RY}
\begin{proposition}
$\Omega$ is random: $\forall n$,
$K(\Omega_{[n]}) > n + O(1)$\,. 
\end{proposition}
\begin{proof} From $\Omega_{[n]}^*$, $\Omega_{[n]}$ can be reconstructed, and by Proposition~\ref{prop:compress}, $\finite{\chi}{2^{n+1}-1}$ can be computed. This, in turn, permits the computation of the output of all programs of length~\mbox{$\leq n$}. All of these strings, therefore, have complexity smaller or equal to $n$. Then, the lexicographically first string that does not have complexity~$\leq n$ can be recognized. Let us denote it $z_n$ (it is $z_B$ with $B=n$).
In diagram form,
\begin{center} 
\begin{picture}(350,85)(-45,-50)
\put(-25,5){\makebox(25,20){$\Omega_{[n]}^*$}} 
\put(0,15){\vector(1,0){15}} 
\put(20,0){\framebox(30,30){$O(1)$}} 
\put(53,15){\vector(1,0){15}} 
\put(73,5){\makebox(25,20){$\finite{\Omega}{n}$}} 
\put(100,15){\vector(1,0){15}} 
\put(120,0){\framebox(30,30){$O(1)$}} 
\put(155,15){\vector(1,0){15}} 
\put(183,5){\makebox(25,20){$\finite{\chi}{2^{n+1}-1}$}} 
\put(218,15){\vector(1,0){15}} 
\put(235,0){\framebox(30,30){$O(1)$}} 
\put(270,15){\vector(1,0){15}} 
\put(-30,-25){\vector(1,0){15}} 
\put(15,-35){\makebox(25,20){$\{x \st K(x) \leq n\}$}} 
\put(70,-25){\vector(1,0){15}} 
\put(90,-40){\framebox(30,30){$O(1)$}} 
\put(125,-25){\vector(1,0){15}} 
\put(202,-35){\makebox(25,20){$ z_n \deq \min \{x \st K(x)>n \}$\,.}} 
\end{picture} \,
\end{center}
This program for $z_n$ is of total size $K(\Omega_{[n]}) + O(1)$, so 
$
K(z_n) \leq K(\Omega_{[n]}) + O(1)
$.
But by definition, $K(z_n) > n$. 
\end{proof}
\pscom{All notions for Problem Sheet \#3 have been covered. See \S\ref{sec:ps3}.}
\smallskip


\subsection{Mathematical Truths in $\Omega$}

The first bits of $\Omega$ encode many mathematical facts.
In particular, they encode the truth or falsity of statements that are \emph{refutable by finite means}, namely, statements that, if false, an algorithm can find a counterexample and halt. 

For instance, the \emph{Goldbach conjecture}, which states that every even number larger than $2$ is the sum of two primes, can be refuted if a counter-example is found. Consider the following program:

\begin{center}
\begin{tabular}{rl}
$q$:&\texttt{For $i = 1,2,3, \ldots$}\\
& \texttt{\qquad check if $2i + 2 = p_1 + p_2$, where $p_1$ and $p_2$ are primes 
}\\
& \texttt{\qquad If yes, go to the next $i$.}\\
& \texttt{\qquad If not, halt.} 
\end{tabular}
\end{center}
Note that $q$ halts if and only if Goldbach's conjecture is false.
Moreover, $q$ is quite short: its halting status is encoded in the first few thousand bits of $\Omega$.

Another example of a statement refutable in finite means is \emph{Fermat's last theorem} (proven by Wiles). It states that the equation
$
x^n + y^n = z^n
$
admits no integer solutions for $x$, $y$, $z$ $> 0$ and $n>2$. An exhaustive search can look for such an integer solution and halt if it finds it.

\emph{The Riemann hypothesis} states that the non-trivial zeros of 
$
\zeta (s) = \sum_{n=1}^\infty \frac{1}{n^s}
$
are all on the line $\re{s} = 1/2$. 
The Riemann hypothesis can be shown to be refutable by finite means, due to a result by Matiyasevich~\cite{matiyasevich1993hilbert}.
Therefore, the first bits of $\Omega$ could tell us if the hypothesis is true or not.

The first bits of $\Omega$ can also be used to determine whether a statement $s$ is provably true ($s$ is a theorem), provably false ($\lnot s$ is a theorem), or independent\footnote{The most well-known examples of independent statements are perhaps Euclid's parallel postulate and the continuum hypothesis.
In Euclidean geometry, \emph{the parallel postulate} states that in a plane, given a line and a point not on it, at most one line parallel to the given line can be drawn through the point; it can be shown to be independent of the other four axioms of Euclid. 
\emph{The continuum hypothesis} states that there is no set whose cardinality is strictly between that of the integers and the real numbers;
it can be shown to be independent of ZFC.} of some formal axiomatic theory $\mathcal A$. 
Remember that
a formal axiomatic theory is a computably enumerable set of mathematical assertions (theorems).
Consider the following program

\begin{center}
\begin{tabular}{rl}
$Q$:& \texttt{Computably enumerate the theorems of $\mathcal A$.} \\
& \texttt{\qquad If $s$ is enumerated halt and output 1.}\\
& \texttt{\qquad If $\lnot s$ is enumerated, halt and output 2.}
\end{tabular}
\end{center}
If~$Q$ loops, $s$ is independent. The program~$Q$ is of finite length (roughly the length of the enumerator for~$\mathcal A$ and for an algorithmic description of $s$). We use~$\Omega$ to determine if~$Q$ halts or loops. If it is promised to halt, we run~$Q$ to find if it outputs~$1$ or~$2$.

 \subsection{Chaitin's Incompleteness Theorem}\label{sec:chaitinincomplete}
  
  \begin{definition}
  We can define the \emph{complexity} $K(\mathcal A)$ of a formal axiomatic theory~$\mathcal A$,  as the length of the shortest program, which computably enumerates the theorems of $\mathcal A$.
  \end{definition}

\begin{theorem}[Chaitin] 
For any formal axiomatic theory~$\mathcal A_0$ there is a constant $k_0$ (which is roughly of size $K(\mathcal A_0)$) such that for all bit strings $x\in \{0,1\}^*$---except a finite number of them---the mathematical assertion:
\smallskip
\ct{``The string~$x$ is such that $K(x) > k_0$''}
\smallskip
is true, but unprovable in~$\mathcal A_0$.
\end{theorem}
\begin{proof}
First, it is true that only finitely many strings have complexity $\leq k_0$. 
Indeed, there is only a finite number of programs shorter than~$k_0$, and each of them can be the shortest program of no more than one string.
Thus, for all other strings, the assertion is true.

\noindent Suppose for contradiction that there is a string whose complexity is \emph{provably} greater than $k_0$. Then consider the following program.
\begin{center}
\begin{tabular}{rl}
$p:~$ &\texttt{Computably enumerate all the proofs of} $\mathcal A_0$ \texttt{until a proof is found of}\\
&\textit{``The string~$y$ is such that $K(y) > k_0$''},  \texttt{for some $y$.}\\
& \texttt{Return that~$y$.}
\end{tabular}
\end{center}
\medskip
The program $p$ can be made of length at most $k_0$, provided that $k_0 \geq K(\mathcal A_0, k_0) + c$.
Yet,~$p$ computes a string $y$ whose \emph{shortest} program has length~$> k_0$. Contradiction.
\end{proof}

For example, the complexity of ZFC might be~$5000$. Then, the overwhelming majority of strings of length~$6000$ bits have a complexity larger than~$5000$, yet this can be proven in ZFC for no such string. 

Compared with Gödel's incompleteness, which might appear to be an anomaly, Chaitin's incompleteness makes it inevitable. It puts information-theoretic limits on the power of formal axiomatic theories. 

\vspace{5pt}
\lecture{AIT-20}{https://youtu.be/I1w_1ASqfD0}\\
\nit{Comment on the lecture: Some solutions to Problem Sheet \#3 are presented.}
\vspace{5pt}

What happens if we enlarge the axioms? 
The complexity of the formal axiomatic theory increases, but it doesn't circumvent the problem.

\begin{theorem}[Chaitin]
There is a constant $c$ such that a formal axiomatic theory~$\mathcal A_0$ with complexity $K(\mathcal A_0)$ can never determine more than $K(\mathcal A_0) + c$ bits of~$\Omega$.
\end{theorem}

\nit{Proof sketch}: Otherwise, we can get a compression of $\Omega$.


Chaitin wrote~\cite{chaitin2023complexity}: 
\begin{quote}
``I believe that complexity reveals the depth of the incompleteness phenomenon: For the Platonic world of mathematical ideas has infinite complexity, but any FAT only has finite complexity. Therefore, any FAT can only capture an infinitesimal part of the richness of pure mathematics.'' 
\end{quote}
He also wrote~\cite{chaitin2023complexity}:
\begin{quote}
``I believe in creativity in mathematics, I believe that math isn’t a closed system, it’s an open system. A fixed formal axiomatic system of the kind that Hilbert wanted for all of math, yes, it would have given us absolute certainty, but it would also have made mathematics into a frozen cemetery.'' 
\end{quote}
He advocates a ``quasi-empirical'' viewpoint on mathematics, where the complexity of formal axiomatic theories can be increased ``by adding new principles, new axioms, that are not self-evident but that are justified by their fertile consequences''~\cite{chaitin2023complexity}. 
%

\vspace{10pt}
\nit{Further reading on the topic:  \cite{bennett1979random, chaitin2007
} \cite[Chapter 8]{deutsch2011beginning}}.\\
\nit{See also} \mbox{\href{https://www.youtube.com/watch?v=CluVy2jICgs&t=2s}{\commeurl{[the conversation I held with Bennett and Deutsch]}}}.

\newpage
\section{The Coding Theorem}
\lecture{AIT-21}{https://youtu.be/28Xt8aeDKYs}
\vspace{5pt}

Reminder: 
\beas
P(x) &=& \sum_{p \st U(p)=x} 2^{-\len p} \\
&=& \underbrace{2^{-\len{x^*}}}_{2^{-K(x)}} + 2^{-\len{p_1}} + 2^{-\len{p_2}} + 2^{-\len{p_3}} + \ldots
\eeas
Clearly $P(x) > 2^{-K(x)}$. But how much larger than $2^{-K(x)}$ can $P(x)$ be? For some $x$ of length $50$, there might be 3 programs of length $K(x)$ and 4 programs of length $K(x) + 1$, so for that $x$, $$P(x) \geq 3 \cdot 2^{-K(x)} + 4 \cdot 2^{-(K(x)+1)} = 5 \cdot 2^{-K(x)}.$$
However, one might expect that as the length of the string increases, the multiplicative factor relating~$P(x)$ and~$2^{-K(x)}$ could grow unboundedly.
For instance, it would suffice to have strings of increasing length whose number of programs of minimal length grows unboundedly.

The coding theorem shows that this is not the case: it states that there exists a constant $c$ such that $\forall x$,
$$
2^{-K(x)} \leq P(x) \leq c \, 2^{-K(x)} \,.
$$
Equivalently,
$
{K(x)} \geq - \log P(x) \geq - \log c + {K(x)}
$.

\begin{definition} The \emph{conditional universal probability} is given by
$$
P(x \given y) = \sum_{p \st U(p,y)=x}2^{-\len p} \,.
$$
\end{definition}


\begin{theorem}[The Coding Theorem]
For all $x$ and for all $y$
$$
K(x \given y) = -\log P(x \given y) + O(1)\,.
$$
\end{theorem}
\begin{proof}
For a first take on the proof, let us set $y=\epsilon$. 
The proof carries through if $y$ sits on the auxiliary tape.

We need to show that $K(x) \leq - \log P(x) + O(1)$.
We shall do this by exhibiting a Turing machine~$T$ that computes each $x$ from an input $a_l$ of length $- \lceil \log P(x) \rceil + 1$.
Since~$U$ can simulate that machine, it provides a program of length \mbox{$- \log P(x) + O(1)$} to compute $x$, and thus the shortest program for $x$ has a length no greater. 

\begin{center}
\begin{tabular}{rl}
\smallskip
$T(a_l):~$ &\texttt{Run all programs, dovetailed fashion.}\\
\smallskip
& \texttt{When the $j$-th program halts (in the dovetail order),}\\
& \texttt{add $(p_j, x_j, S_{x_j}, a_j)$ to a list of quadruples, where}\\
& \texttt{\qquad $p_j$ is the program that halted}\\
& \texttt{\qquad $x_j = U(p_j)$ is its output}\\
& \texttt{\qquad $S_{x_j} = \sum_{i \leq j \st x_i = x_j} 2^{-\len{p_i}}$ is the current value}\\
& \texttt{\qquad \qquad of the lower semi-computation of $P(x_j)$ and}\\
& \texttt{\qquad $a_j$ is $\emptyset$ if the position of the leading $1$}\\
& \texttt{\qquad \qquad in the binary expansion of $S_{x_j}$ has not changed}\\
& \texttt{\qquad \qquad and if it has changed, 
$a_j$  is the address of} \\
& \texttt{\qquad \qquad the lexicographically first available node in}\\
& \texttt{\qquad \qquad the binary tree at depth $\lceil -\log S_{x_j} \rceil + 1$}\\
& \texttt{Output the string that is allocated on the node $a_l$.} 
\end{tabular}
\end{center}
\medskip

 When $a_j \neq \emptyset$, the idea is to allocate the string $x_j$ to the node $a_j$. To preserve a prefix-free tree, when a node is allocated, all its decedent branches are removed from the tree.

Note that he quantity $\lceil - \log \alpha \rceil$ corresponds to the position of the first $1$ in the binary expansion of $\alpha$. Indeed, let 
$$\alpha = \underbrace{0.000\ldots01b_{i+1}b_{i+2}\ldots}_{\text{First $1$ is at the $i$-th bit after the point}} = 2^{-i} + b_{i+1} 2^{-(i+1)}  + b_{i+2} 2^{-(i+2)} + \ldots \,.$$
$\alpha \in [2^{-i}, 2^{-i+1}) 
\implies -\log \alpha \in ({i-1}, {i}] 
\implies
\lceil - \log \alpha \rceil = i \,.$

\begin{example}
Let $x = 11011$. 
See Figure~\ref{fig:excoding}

\begin{figure}[H]
 \begin{minipage}{0.45\textwidth}
\begin{center}
\begin{tabular}{ c|c|c|c|c } 

 $i$ & $p_i$& $x_i$ & $S_{x_i}$ & $a_i$ \\ 
 \hline
1 &  \scriptsize{1100} &  \scriptsize{01} &  \scriptsize{0.0001} &   \scriptsize{00000}\\
  \hline
 2 & \scriptsize{00110} & \scriptsize{11011} & \scriptsize{$0.00001$} &  \scriptsize{000010}\\
  \hline
 3 & \scriptsize{000} & \scriptsize 1 & \scriptsize{$0.001$} &  \scriptsize{0001} \\
  \hline
 4 & \scriptsize{011101} & \scriptsize{11011} & \scriptsize{$0.000011$} &  \scriptsize{$\emptyset$} \\
 \hline
  5 & \scriptsize{111} & \scriptsize{11011} & \scriptsize{$0.001011$} &  \scriptsize{0010} \\
  \hline
  \dots &  \dots &  \dots &  \dots &  \dots 
\end{tabular}
\end{center}
\end{minipage}
\begin{minipage}{0.45\textwidth}
\begin{tikzpicture}
\node[circle, draw, fill=black, inner sep=0.5pt, name=root] at (0,0) {};

\node[circle, draw, fill=black, inner sep=0.5pt, name=n0] at (1,0.5) {};
\node[name=n1] at (1,-0.5) {};
\draw (root) -- node[midway, above] {\scriptsize 0} (n0);
\draw (root) -- (n1);

\node at (0.5, -0.0) {\scriptsize 1};
\node[right=-8pt of n1] {$\dots$};

\node[circle, draw, fill=black, inner sep=0.5pt, name=n00] at (2,1) {};
\node[name=n01] at (2,0) {};
\node[right=-8pt of n01] {$\dots$};
\draw (n0) -- node[midway, above] {\scriptsize 0} (n00);
\draw (n0) -- (n01);

\node at (1.5, 0.5) {\scriptsize 1};

\node[circle, draw, fill=black, inner sep=0.5pt, name=n000] at (3,1.5) {};
\node[circle, draw, fill=black, inner sep=0.5pt, name=n001] at (3,0.5) {};
\draw (n00) -- node[midway, above] {\scriptsize 0} (n000);
\draw (n00) -- node[midway, above] {\scriptsize 1} (n001);

\node[circle, draw, fill=black, inner sep=0.5pt, name=n0000] at (4, 2) {};
\node[circle, draw, fill=black, inner sep=1.5pt, name=n0001] at (4, 1.25) {};
\draw (n000) -- node[midway, above] {\scriptsize 0} (n0000);
\draw (n000) -- (n0001);

\node at (3.5, 1.55) {\scriptsize 1};

\node[circle, draw, fill=black, inner sep=1.5pt, name=n0010] at (4, 0.75) {};
\node[name=n0011] at (3.5, 0.375) {};
\draw (n001) -- node[midway, above] {\scriptsize 0} (n0010);
\draw (n001) -- (n0011);
\node[right=-8pt of n0011] {$\dots$};

\node[right=-1pt of n0001] {$1$};
\node[right=-1pt of n0010] {$11011$};

\node[circle, draw, fill=black, inner sep=1.5pt, name=n00000] at (5, 2.5) {};
\node[circle, draw, fill=black, inner sep=0.5pt, name=n00001] at (5, 1.5) {};
\draw (n0000) -- node[midway, above] {\scriptsize 0} (n00000);
\draw (n0000) -- node[midway, above] {\scriptsize 1} (n00001);
\node[right=-1pt of n00000] {$01$};

\node[circle, draw, fill=black, inner sep=1.5pt, name=n000010] at (6, 2) {};
\node[name=n000011] at (5.5, 1.25) {};
\node[right=-8pt of n000011] {$\dots$};
\draw (n00001) -- node[midway, above] {\scriptsize 0} (n000010);
\draw (n00001) -- (n000011);
\node[right=-1pt of n000010] {$11011$};
\end{tikzpicture}
\end{minipage}
\caption{Suppose that $T$ gives an enumeration of halting programs as in the table on the left. The first halting program is of length $4$, so its output, $01$, is assigned to the lexicographically first available node of length $5$. That node is $00000$. Then, the string $11011$ is computed by a $5$-bit program, so $11011$ is assigned to the node $000010$. The string $11011$ is computed later by a $6$-bit program, but adding $2^{-6}$ to $S_{11011}$, which lower semi-computes $P(11011)$, does not in this case change the position of the leading $1$, so no node is assigned. But when a 3-bit program later computes also $11011$, the string is allocated to a node at depth $4$.
}
\label{fig:excoding}
\end{figure}
Suppose that
$P(11011) = 0.00110111...$.
The lower semi-computation of $S_{11011}$ would continue forever, but $11011$ will not be assigned to more nodes, because the leading $1$ of $S_{11011}$ has been stabilized.

%
%
The machine $T$ can compute~$x$ via the massive computation that builds the tree for all strings, together with the address $0010$, which tells it to stop and output the string allocated to that address. Here, this is $x=11011$.
\end{example}

In general, because $S_x$ lower semi-computes $P(x)$, there will eventually be a program $p_l$ that computes $x$ ($x=x_l$) and whose contribution to $S_x$ will make the leading $1$ in the binary expansions of $S_x$ to be at the same position as the leading $1$ in the binary expansion of $P(x)$.
This will give a first (and only) node at depth~$\lceil -\log S_{x_l} \rceil +1 = \lceil -\log P(x) \rceil + 1$ to which the string $x_l = x$ will be allocated. The address of this node is $a_l$, and $\len{a_l} =  \lceil -\log P(x) \rceil + 1$.

We should verify that this way of allocating strings to nodes in the binary tree will never run out of space.

\beas
\sum_{j \st a_j \neq \emptyset} 2^{-\len{a_j}} 
&=& \sum_x \sum_{\substack{j \st a_j \neq \emptyset \\ \text{and } x_j=x}} 2^{-\len{a_j}}\\
&<& \sum_x 2^{-(\lceil -\log P(x) \rceil +1)} + 2^{-(\lceil -\log P(x) \rceil +2)} + 2^{-(\lceil -\log P(x) \rceil +3)} + \ldots \\
&=& \sum_x  2^{\lfloor \log P(x) \rfloor}2^{-1} \left(1+ \frac 12 + \frac 14 + \ldots \right)\\
&\leq& \sum_x P(x)\\
&\leq& 1 \,.
\eeas
The third line uses the fact that $- \lceil - \beta \rceil = \lfloor \beta \rfloor$.


By Kraft's inequality, the $|a_j|$'s can be the lengths prefix-free set of codewords, and the $a_j$'s are codewords.


Thus, a program for $U$ to produce $x$ is $\bar i a_l$, where $T\equiv T_i$. It has length $- \log P(x) + O(1)$.
\end{proof}

\lecture{AIT-22}{https://youtu.be/-JEWx441j7A}
\vspace{5pt}

A subtlety: The domain of the definition of $T$ above is a prefix-free set of strings, fine, but is $T$ a self-delimiting Turing machine? The concern is that $U$ can simulate all self-delimiting Turing machines.


\begin{proposition} \label{prop:abstract}
Let $y$ be some auxiliary string. Let $T_{\text{prefix}}$ be a Turing machine whose domain~$\{a \st T_{\text{prefix}}(a,y) \halts \}$ is a prefix-free set. Then there exists a self-delimiting Turing machine $T_{\text{s-d}}$ s.t. 
$T_{\text{prefix}}(\cdot, y) \equiv T_{\text{s-d}} (\cdot, y)$.
\end{proposition}
\begin{proof}
The idea is that $T_{\text{s-d}}$ should read another square of its program tape only when it is necessary.
Suppose $T_{\text{s-d}}$ has the string $y$ on its work/auxiliary tape (no self-delimitation constraints are required on the auxiliary information).
$T_{\text{s-d}}$ generates the computably enumerable set $S = \{s \st T_{\text{prefix}}(s, y) \halts\}$ on its work tape.
As it generates $S$, $T_{\text{s-d}}$ continually checks whether or not the part $p$ of the program that it has already read (initially~$p = \epsilon$) is a prefix of some known element~$s$ of~$S$. 
Whenever~$T_{\text{s-d}}$ finds that $p$ is a prefix of an $s \in S$, it does the following. 
If $p$ is a proper prefix of $s$ (i.e. if $s=pw$, for $w \neq \epsilon$), $T_{\text{s-d}}$ reads another square of the program tape. This is because it knows that $p$ cannot be in the domain, otherwise the domain would not be prefix-free. 
And if $p = s$, $T_{\text{s-d}}$ computes and outputs $T_{\text{prefix}}(p, y)$ on the output tape.
\end{proof}

\subsection{Lower Semi-computable Semi-measures}\label{sec:lss}

\begin{definition}
A \emph{discrete semi-measure} is a map
$
\mu:\mathbb N \to \mathbb R
$
such that $$\sum_{x=1}^\infty \mu(x) \leq 1\,.$$ 
\end{definition}
\begin{definition}
The semi-measure is \emph{lower semi-computable} if there exists a partial computable function~$\tilde\mu(x,t)$ such that~$\tilde \mu(x,t) \leq \tilde \mu(x,t+1)$ and~$\lim_{t \to \infty} \tilde \mu(x,t) = \mu(x)$.
\end{definition}
\begin{definition}
The semi-measure $m$ is \emph{universal} with respect to a class $\mathcal C$ of semi-measures if it is in that class and it is ``larger than any other'', namely, if there exists a universal constant $c$ such that for all $\mu \in \mathcal C$,
$$
\mu(x) \leq c \, m(x) \qquad \forall x \,.
$$
\end{definition}

\begin{definition} Let $T$ be a self-delimiting Turing machine.
$$
P_T(x) = \sum_{p \st T(p)= x}2^{-\len p} \,.
$$
\end{definition}

\begin{proposition}
Let $\mu$  be a lower semi-computable semi-measure. There exists a self-delimiting Turing machine $T$ such that $\mu(x) = P_T(x)$ for all $x$.
\end{proposition}

\begin{proof}
The machine~$T$ calls the lower semi-computation of $\mu$, with initialized value~$\tilde \mu(x,0)= 0$.
For $t = 1, 2, \ldots$, run $\tilde \mu(x,t)$ for all $x\leq t$. When some $\tilde \mu(x,t)$ is produced with $\tilde \mu(x,t)-\tilde \mu(x,t-1) = \ell > 0$ (an update on the lower bound of $\mu(x)$), assign to $x$ the subset of $[0,1]$ which is the next available interval of length $\ell$ (starting from 0). 

Let~$S_x$ be the union of all the intervals assigned to~$x$.
Note that the total length of the intervals that compose~$S_x$ converges to~$\mu(x)$, and since~$\mu$ is a semimeasure,~$[0,1]$ is long enough to accomodate all~$\{S_y\}_{y}$.

From this construction of~$S_x$ we define programs for~$T$ which produce~$x$. The set~$S_x$ can be covered from within by cylinders $\Gamma_p = [0.p, 0.p + 2^{-\len p})$, namely, $\Gamma_p \subseteq S_x$ and $\cup_p \Gamma_p = S_x$. The $p$'s are $T$'s programs for $x$.
\end{proof}

\begin{remark}
This proof provides another equivalent ``probabilistic'' interpretation of the input tape. One can view the input tape as already containing infinitely many bits (instead of calling them by the monkey one by one). All infinite strings $s$ which would yield a computation of $x$ by $T$ are of the form $s=ps'$ with $T(p)=x$. The Lebesgue measure of all such $s$ is $P_T(x)$.
\end{remark}


\begin{proposition}
$P(x)$ is a universal semi-measure for the class of lower semi-computable semi-measures.
\end{proposition}
\begin{proof}
By the last Proposition, it suffices to show that $P(x) \deq P_U(x)$ dominates all $P_T(x)$. 
By the universality of $U$, there exists $r$ such that $U(rq)=T(q)$. 
Thus
\beas
P_T(x) &=&
 2^{|r|}2^{-|r|}\sum_{q \st T(q)=x} 2^{-\len q} \\
&=& 2^{|r|} \sum_{q \st U(rq)=x} 2^{-\len{rq}} \\
&\leq& 2^{|r|} \sum_{p \st U(p)=x} 2^{-\len{p}} \\
&=& 2^{|r|} P(x) \,.
\eeas
\end{proof}

\begin{corollary}
By the coding theorem, $2^{-K(x)}$ multiplicatively dominates $P(x)$, and therefore, $2^{-K(x)}$ is also a universal lower semi-computable semi-measure.
\end{corollary}

\begin{example}
\label{excr}
The mapping $$y \mapsto \sum_x 2^{-K(x,y)}$$ is a lower semi-computable semi-measure.
Indeed, it is lower semi-computable:
On input $y$,
for $i= 1,2, \dots$ run all $p$ of length $\leq i$. When some $p$ halts with $\langle x, y \rangle$ for some $x$, add $2^{-\len p}$. To the sum. When some $p'$ halts later with $\langle x,y \rangle$ (the same $x$ as before), add $2^{-\len{p'}} - 2^{- \len p}$ if $p'$ is shorter than $p$. Otherwise, don't add anything.

It is a semi-measure, because by Kraft's inequality,
$$
\sum_y \sum_x 2^{-K(x,y)} \leq 1 \,.
$$

Thus there is a $c$ st $\forall y$ $\sum_x 2^{-K(x,y)} \leq c 2^{-K(y)}$, or equivalently,
\be \label{eqexcr}
\frac{\sum_x 2^{-K(x,y)}}{2^{-K(y)}} \leq c \,.
\ee
\end{example}

\begin{proposition}[Chain rule, hard side] \label{prop:crhard}
For all $x$ and $y$,
$$
K(y) + K(x \given y^*) \leq K(x,y) + O(1) \,.
$$
\end{proposition}
\begin{proof}
The statement to prove is $\exists c$ s.t.
$$
K(y) + K(x \given y^*) \leq K(x,y) + c
$$
$$
\iff
$$
$$
-K(x,y) + K(y)  \leq - K(x \given y^*)  + c
$$
$$
\iff
$$
$$
\frac{2^{-K(x,y)}}{2^{-K(y)}} \leq \underbrace{2^c}_{c'} 2^{- K(x \given y^*)}  \,.
$$
This will be obtained if we show that $x \mapsto \frac{2^{-K(x,y)}}{2^{-K(y)}}$ is a lower semi-computable semi-measure given $y^*$. 

\begin{enumerate}
\item It is lower semi-computable. Since we have $y^*$, we have $K(y) = \len{y^*}$, so we do not have to worry about the denominator. 
The string $y$ can be obtained from $y^*$, and since $x$ is given, we need to show that $2^{-K(x,y)}$ is lower semi-computable given $x,y$.

Run all $p$ dovetail fashion, when some $p$ is found for which $U(p)= \langle x, y \rangle$, update the lower bound to $2^{-\len p}$.
The smaller the length of the $p$, the higher the lower bound. Eventually, the witness of $K(x,y)$ is found.


\item Up to a renormalization which introduces a multiplicative constant, is a semi-measure. See Equation~\eqref{eqexcr} of Example~\ref{excr}. 
\end{enumerate}

\end{proof}
Analogously with the probabilistic setting (\S\ref{sec:mi}), the $O(1)$-chain rule entails a symmetric notion of \emph{algorithmic mutual information},
\beas
I(x;y) &\deq& K(x) - K(x \given y)\\
&\stackrel{\scriptsize{\text{c.r.}}}{=}& K(x) - K(x,y) + K(y) + O(1)\\
&\stackrel{\scriptsize{\text{c.r.}}}{=}& K(y) - K(y \given x) + O(1)\\
&\deq& I(y;x) + O(1)\,.
\eeas 

\pscom{All notions for Problem Sheet \#4 have been covered. See \S\ref{sec:ps4}.}

\newpage
\appendix
\section{Problem Sheets}
\subsection{Problem Sheet \#1}\label{sec:ps1}
\bigskip

\noindent 1. We identified strings with natural numbers via $\left\{ (\epsilon, 1), (0, 2), (1,3), (00,4) \ldots \right\}$, and we defined $\bar x = 1^{|x|}0x$ and $x'=  \overline{|x|}x$.
Find $y$, $z$ and $t$ if
$$
\bar y z' t= 111110001011110000110100101111100100101 \,.
$$

\bigskip
\noindent 2. It was once asked (on 23.02.23, in D1.14) whether there are countably many or uncountably many prefix codes of a countable alphabet. Answer that question.
\emph{Hint:} subsets of $\{0, 10, 110, \ldots \}$.

\bigskip
\noindent 3. A conjecture was once made (on 23.02.23, in D1.14) that a set of strings is both prefix-free and suffix-free if and only if the strings have constant length. Disprove the conjecture by exhibiting an infinite prefix- and suffix-free set of binary strings.


\bigskip
\noindent 4. Does there exist a prefix code $\Phi: \{0,1\}^* \backslash \{\epsilon\} \to \{0,1\}^*$ such that 
$$|\Phi (x)| = |x| + \lceil \log |x| \rceil + 1 \qquad \forall x \,?$$
Prove your claim.

\bigskip
\noindent 5. \emph{Shannon-Fano code}. Let $X$ be a random variable over the alphabet $\mathcal X$.
Explain how to construct a prefix code of $\mathcal X$ with expected code-word length
$
 < H(X) + 1
$.


\bigskip
\noindent 6. Prove the chain rule for entropy:
$
H(X,Y) = H(X) + H(Y \given X) \,.
$

\newpage

\subsection{Problem Sheet \#2}\label{sec:ps2}
\bigskip

\noindent 
We defined an \emph{enumerator}
of a set $A \subseteq \mathbb N$ as a Turing machine that enumerates all (and only) elements of $A$ \emph{without order}\footnote{The order is given by $\mathbb N$.} and \emph{possibly with repetitions}.

\bigskip
\noindent 1. A \emph{non-repeating enumerator} 
of a set $A$ is a Turing machine that enumerates all (and only) elements of $A$ \emph{without order} and \emph{without repetitions}.
Complete the following statement \mbox{with either ``computable'' or ``computably enumerable,'' and prove it.}
\begin{center}
A set $A$ can be enumerated by a non-repeating enumerator
\\
 if and only if it is \underline{~~~~~~~~~~~~~~~~~~~~~~~~~~~~~~~~~~~~~~~~~~}. 
\end{center}

\bigskip
\noindent 2. An \emph{ordered enumerator}
of a set $A$ is a Turing machine that enumerates all (and only) elements of $A$ \emph{in order} and \emph{possibly with repetitions}.
Complete the following statement with either ``computable'' or ``computably enumerable,'' and prove it.
\begin{center}
An infinite set $A$ can be enumerated by an ordered enumerator
\\
 if and only if it is \underline{~~~~~~~~~~~~~~~~~~~~~~~~~~~~~~~~~~~~~~~~~~}. 
\end{center}

\bigskip
\noindent 3. Prove that the halting set $K_0 = \{\langle i,j \rangle \,|\, \varphi_i(j) \halts\} \,$ is computably enumerable.

\bigskip
\noindent 4. We saw that the set of indices $i$ such that the range of $\varphi_i$ is nonempty is computably enumerable. It was asked (on 28.03.23) if the set is also computable. Answer that question and provide a brief justification.

\bigskip
\noindent 5. Prove that a set $A \subseteq \mathbb N$ is computable if and only if both $A$ and $\overline A = \mathbb N \setminus A$ are computably enumerable.

\smallskip
\noindent \emph{Comment.} The class of all computable sets\footnote{In complexity theory, a subset of $\{0,1\}^* \simeq \mathbb N$ is also called a \emph{language}.} is denoted $\textbf{R}$, the class of all computably enumerable sets is denoted $\textbf{RE}$, and the class of all sets whose complement is computably enumerable is denoted $\text{co}\textbf{RE}$. Question 5 shows that
$$
\textbf{R} = \textbf{RE} \cap \text{co}\textbf{RE} \,.
$$

\newpage

\subsection{Problem Sheet \#3}\label{sec:ps3}
\bigskip

\noindent \emph{Note 1.} One of these questions is a bonus: 6 correct solutions yield 60/60, 7 correct solutions yield 70/60.

\noindent \emph{Note 2.} If some proof involves a program in the setting of prefix complexity, explain what makes the program self-delimited.

\bigskip
\noindent 1. Let $0^n \deq \underbrace{00\ldots0}_{n \text{ times}}$. Show that $C(0^n) = C(n) + O(1)$.

\bigskip
\noindent 2. Show that for all $x \in \{0,1\}^n$, $K(x) = C(x) + O(\log n)$.

\bigskip
\noindent 3. Prove the easy side of the chain rule for prefix complexity: For all $x$ and $y$,
$$
K(x,y) \leq K(x) + K(y \given x^*) + O(1) \,.
$$

\bigskip
\noindent 4. Show that there exist infinitely many strings such that 
$$
K(x) > \log x + \log \log x + \log \log \log x \,.
$$
\noindent \emph{Hint 1:} $\sum_x 2^{-K(x)}$. \\
\noindent \emph{Hint 2:} $\frac{d}{dx} \log \log \log x = \frac{1}{\log \log x} \frac{1}{\log x} \frac{1}{ x}$.

\bigskip
\noindent 5. Let $A$ be a finite set whose cardinality is denoted $|A|$. 
Find an appropriate definition of $K(A)$, and use your definition to show that for all $x \in A$,
$$
K(x) \leq K(A) + \log |A| + O(1) \,.
$$

\bigskip
\noindent 6. \emph{$\Omega$'s little cousin.}
Let $\omega_n$ be the cardinality of the following set
$$
\{p : U(p) \searrow ~~\text{and}~~ |p| \leq n\}\,.
$$
Show that $K(\omega_n, n) > n + O(1)$\,.

\bigskip
\noindent 7. A function is \emph{upper semi-computable} if there exists a never-ending algorithm which finds ever smaller upper bounds to the graph of the function, and its tentative upper bounds converge to the graph of the function in the limit of infinite time. Show that the function 
\beas
K : \mathbb N &\to& \mathbb N \\
x &\mapsto& K(x)
\eeas
is upper semi-computable.

\newpage
\subsection{Problem Sheet \#4}\label{sec:ps4}
\bigskip

\noindent \emph{Notes.} Question 3 is only for fun. No points, no bonus. Just fun.
Question 6 is worth as many points as one chooses to solve the easier or the harder version.

\bigskip
\noindent 1. \emph{The incompressibility method.} Use the existence of incompressible strings to show that there exist infinitely many primes.

\bigskip
\noindent 2. Let $x$ be an $n$ bit string with exactly as many $0$'s as $1$'s. Using either the plain or the prefix setting, show that $x$ is compressible given $n$.

\smallskip
\noindent \emph{Stirling's formula:} $n! = n^n e^{-n} \sqrt{2 \pi n} \left(1 + O(\frac 1n) \right)$.

\bigskip
\noindent 3. \emph{Only for fun.} ``But how can the result of question 2~be possible? 
By the law of large numbers, in a probabilistic random string, we should expect as many 0's as 1's.''
Yes, but not \emph{exactly} as many, for this is a pattern which entails compression. In fact, the central limit theorem states that we should expect variations of about~$\sqrt n$.

\smallskip
\noindent \emph{Easier.} Show that the analysis performed in question 2~fails if the number of $1$'s is $\frac n2 + k$, where $k = O(\sqrt n)$ and $k$ is an incompressible number given $n$.

\smallskip
\noindent \emph{Harder.} Show that $x$ is incompressible given $n$ only if the number of $1$'s in $x$ is $\frac n2 + k$, where $k = O(\sqrt n)$ and $k$ is an incompressible number given $n$.

\bigskip
\noindent 4. A program is \emph{elegant} if it is the shortest program that computes some string. Show that a formal axiomatic theory $\mathcal A$ cannot prove the elegance of programs of length significantly larger than $K(\mathcal A)$.

\bigskip
\noindent \emph{Note.} That's another way to incompleteness: Only a finite number of programs can be proven elegant, yet there are infinitely many elegant programs.

\bigskip
\noindent 5. a) Write down a function that grows extremely fast. 
The faster it grows, the better. 
\emph{Don't read further before having written down a function.}

\smallskip
\noindent b) Read online about the \emph{Ackermann} function. Does it grow faster than yours?

\smallskip
\noindent c) The \emph{Busy Beaver} function is defined as
$$
B(n) \deq \max \{Rt(p) ~:~ U(p) \halts ~~\&~~\len p \leq n\}\,,
$$
where~$Rt(p)$ denotes the \emph{run time} of $p$, namely, the number of computational steps before $U(p)$ halts.
Show that $B(n)$ grows faster than any computable function, namely, show that for all computable $f$, there exists $n_0$ such that for all $n\geq n_0$, $B(n) > f(n)$.

\smallskip
\noindent \emph{Note.} Assuming that your function was computable, the Busy Beaver beats both your and Ackermann's functions.

\smallskip
\noindent d) Show that $\Omega_{[n]}$ permits to compute $B(n)$.


\bigskip
\noindent 6. \emph{Lower semicomputable semimeasures.} Let
$$
\omega'_n \deq \left |\{p \,:\, U(p) \searrow ~\text{and}~|p| = n\} \right| \qquad \text{and} \qquad
\omega_n \deq \left | \{p \,:\, U(p) \searrow ~\text{and}~ |p| \leq n\} \right | \,.
$$
Answer only one version of the following questions.

\smallskip
\noindent \,\emph{Easier}.
Show that $\omega_n' \leq 2^{n-K(n)+O(1)}$.

\smallskip
\noindent \,\emph{Harder}.
Show that $\omega_n \leq 2^{n-K(n)+O(1)}$.

\bigskip
\noindent 7. Write down a question that you have regarding any topic related to AIT. If you would like to comment on why you find your question intriguing, please do it. \\
If you have ideas on how to solve your question, please elaborate.

\newpage
\section{Proposed Articles for the Seminar}\label{sec:papers}

\subsubsection*{Theoretical, Foundational}

\normalsize \noindent 1. 
Gregory~J Chaitin.
\newblock Incompleteness theorems for random reals.
\newblock {\em Advances in Applied Mathematics}, 8(2):119--146, 1987.
\begin{quote}
\scriptsize \nit{Abstract}: 
We obtain some dramatic results using statistical mechanics-thermodynamics kinds of arguments concerning randomness, chaos, unpredictability, and uncertainty in mathematics. We construct an equation involving only whole numbers and addition, multiplication, and exponentiation, with the property that if one varies a parameter and asks whether the number of solutions is finite or infinite, the answer to this question is indistinguishable from the result of independent tosses of a fair coin. This yields a number of powerful Gödel incompleteness-type results concerning the limitations of the axiomatic method, in which entropy-information measures are used.
\end{quote}

\normalsize \noindent 2. 
Markus M{\"u}ller.
\newblock Stationary algorithmic probability.
\newblock {\em Theoretical Computer Science}, 411(1):113--130, 2010.
\begin{quote}
\scriptsize \nit{Abstract :}
Kolmogorov complexity and algorithmic probability are defined only up to an additive resp. multiplicative constant, since their actual values depend on the choice of the universal reference computer. In this paper, we analyze a natural approach to eliminate this machine-dependence.
Our method is to assign algorithmic probabilities to the different computers themselves, based on the idea that “unnatural” computers should be hard to emulate. Therefore, we study the Markov process of universal computers randomly emulating each other. The corresponding stationary distribution, if it existed, would give a natural and machine-independent probability measure on the computers, and also on the binary strings.
Unfortunately, we show that no stationary distribution exists on the set of all computers; thus, this method cannot eliminate machine-dependence. Moreover, we show that the reason for failure has a clear and interesting physical interpretation, suggesting that every other conceivable attempt to get rid of those additive constants must fail in principle, too.
However, we show that restricting to some subclass of computers might help to get rid of some amount of machine-dependence in some situations, and the resulting stationary computer and string probabilities have beautiful properties.
\end{quote}

\subsubsection*{Algorithmic statistics}

 \noindent 3. 
Nikolay Vereshchagin and Alexander Shen.
\newblock Algorithmic statistics: {F}orty years later.
\newblock In {\em Computability and Complexity}, pages 669--737. Springer,
  2017.
\begin{quote}
\scriptsize \nit{Abstract}: Algorithmic statistics has two different (and almost orthogonal) motivations. From the philosophical point of view, it tries to formalize how the statistics works and why some statistical models are better than others. After this notion of a "good model" is introduced, a natural question arises: it is possible that for some piece of data there is no good model? If yes, how often these bad (``non-stochastic'') data appear "in real life"? 
Another, more technical motivation comes from algorithmic information theory. In this theory a notion of complexity of a finite object ($=$amount of information in this object) is introduced; it assigns to every object some number, called its algorithmic complexity (or Kolmogorov complexity). Algorithmic statistic provides a more fine-grained classification: for each finite object some curve is defined that characterizes its behavior. It turns out that several different definitions give (approximately) the same curve. 
In this survey we try to provide an exposition of the main results in the field (including full proofs for the most important ones), as well as some historical comments. We assume that the reader is familiar with the main notions of algorithmic information (Kolmogorov complexity) theory.
\end{quote}

\normalsize \noindent 4. 
Paul~MB Vit{\'a}nyi and Ming Li.
\newblock Minimum description length induction, bayesianism, and kolmogorov
  complexity.
\newblock {\em IEEE Transactions on information theory}, 46(2):446--464, 2000.
\begin{quote}
\scriptsize \nit{Abstract}: The relationship between the Bayesian approach and the minimum description length approach is established. We sharpen and clarify the general modeling principles minimum description length (MDL) and minimum message length (MML), abstracted as the ideal MDL principle and defined from Bayes's rule by means of Kolmogorov complexity. The basic condition under which the ideal principle should be applied is encapsulated as the fundamental inequality, which in broad terms states that the principle is valid when the data are random, relative to every contemplated hypothesis and also these hypotheses are random relative to the (universal) prior. The ideal principle states that the prior probability associated with the hypothesis should be given by the algorithmic universal probability, and the sum of the log universal probability of the model plus the log of the probability of the data given the model should be minimized. If we restrict the model class to finite sets then application of the ideal principle turns into Kolmogorov's minimal sufficient statistic. In general, we show that data compression is almost always the best strategy, both in model selection and prediction.
\end{quote}

\subsubsection*{Induction}

\normalsize \noindent 5. 
Samuel Rathmanner and Marcus Hutter.
\newblock A philosophical treatise of universal induction.
\newblock {\em Entropy}, 13(6):1076--1136, 2011.
\begin{quote}
\scriptsize \nit{Abstract}:
Understanding inductive reasoning is a problem that has engaged mankind for thousands of years. This problem is relevant to a wide range of fields and is integral to the philosophy of science. It has been tackled by many great minds ranging from philosophers to scientists to mathematicians, and more recently computer scientists. In this article we argue the case for Solomonoff Induction, a formal inductive framework which combines algorithmic information theory with the Bayesian framework. Although it achieves excellent theoretical results and is based on solid philosophical foundations, the requisite technical knowledge necessary for understanding this framework has caused it to remain largely unknown and unappreciated in the wider scientific community. The main contribution of this article is to convey Solomonoff induction and its related concepts in a generally accessible form with the aim of bridging this current technical gap. In the process we examine the major historical contributions that have led to the formulation of Solomonoff Induction as well as criticisms of Solomonoff and induction in general. In particular we examine how Solomonoff induction addresses many issues that have plagued other inductive systems, such as the black ravens paradox and the confirmation problem, and compare this approach with other recent approaches.
\end{quote}

\normalsize \noindent 6. 
Marcus Hutter.
\newblock A theory of universal artificial intelligence based on algorithmic
  complexity.
\newblock {\em arXiv preprint cs/0004001}, 2000.
\begin{quote}
\scriptsize \nit{Abstract}: Decision theory formally solves the problem of rational agents in uncertain worlds if the true environmental prior probability distribution is known. Solomonoff’s theory of universal induction formally solves the problem of sequence prediction for unknown prior distribution. We combine both ideas and get a parameterless theory of universal Artificial Intelligence. We give strong arguments that the resulting AI$\xi$ model is the most intelligent unbiased agent possible. We outline for a number of problem classes, including sequence prediction, strategic games, function minimization, reinforcement and supervised learning, how the AI$\xi$ model can formally solve them. The major drawback of the AI$\xi$ model is that it is uncomputable. To overcome this problem, we construct a modified algorithm AI$\xi$tl, which is still effectively more intelligent than any other time $t$ and space l bounded agent. The computation time of AI$\xi$tl is of the order $t \cdot 2^l$. Other discussed topics are formal definitions of intelligence order relations, the horizon problem and relations of the AI$\xi$ theory to other AI approaches.
\end{quote}

\subsubsection*{Nature and Computation}

 \noindent 7. 
Charles~H Bennett.
\newblock Logical depth and physical complexity.
\newblock {\em The Universal Turing Machine A Half-Century Survey}, pages
  227--257, 1988.
\begin{quote}
\scriptsize \nit{Abstract}: {Some mathematical and natural objects (a random sequence, a sequence of zeros, a perfect crystal, a gas) are intuitively trivial, while others (e.g. the human body, the digits of $\pi$) contain internal evidence of a nontrivial causal history.
We formalize this distinction by defining an object’s “logical depth” as the time required by a standard universal Turing machine to generate it from an input that is algorithmically random (i.e. Martin-Löf random). This definition of depth is shown to be reasonably machine-independent, as well as obeying a slow-growth law: deep objects cannot be quickly produced from shallow ones by any deterministic process, nor with much probability by a probabilistic process, but can be pro- duced slowly.
Next we apply depth to the physical problem of “self-organization,” inquiring in particular under what conditions (e.g. noise, irreversibility, spatial and other symmetries of the initial conditions and equations of motion) statistical-mechanical model systems can imitate computers well enough to undergo unbounded increase of depth in the limit of infinite space and time.}
\end{quote}

\normalsize \noindent 8. 
J{\"u}rgen Schmidhuber.
\newblock The fastest way of computing all universes.
\newblock In {\em A computable universe: Understanding and exploring nature as
  computation}, pages 381--398. World Scientific, 2013.
\begin{quote}
\scriptsize \nit{Abstract}: Is there a short and fast program that can compute the precise history of our universe, including all seemingly random but possibly actually deterministic and pseudo-random quantum fluctuations? There is no physical evidence against this possibility. So let us start searching! We already know a short program that computes all constructively computable uni- verses in parallel, each in the asymptotically fastest way. Assuming ours is computed by this optimal method, we can predict that it is among the fastest compatible with our existence. This yields testable predictions. Note: This paper extends an overview of previous work presented in a survey for the German edition of Scientific American.
\end{quote}

\subsubsection*{Thermodynamics}

\normalsize \noindent 9. 
Paul Vit{\'a}nyi and Ming Li.
\newblock Algorithmic arguments in physics of computation.
\newblock In {\em Algorithms and Data Structures: 4th International Workshop,
  WADS'95 Kingston, Canada, August 16--18, 1995 Proceedings 4}, pages 315--333.
  Springer, 1995.
\begin{quote}
\scriptsize \nit{Abstract} :
We show the usefulness of incompressibility arguments based on Kolmogorov complexity in physics of computation by several examples. These include analysis of energy parsimonious 'adiabatic' computation, and scalability of network architectures.
\end{quote}

\normalsize \noindent 10. 
Wojciech~H Zurek.
\newblock Algorithmic randomness and physical entropy.
\newblock {\em Physical Review A}, 40(8):4731, 1989.
\begin{quote}
\scriptsize \nit{Abstract} :
Algorithmic randomness provides a rigorous, entropylike measure of disorder of an individual, microscopic, definite state of a physical system. It is defined by the size (in binary digits) of the shortest message specifying the microstate uniquely up to the assumed resolution. Equivalently, al- algorithmic randomness can be expressed as the number of bits in the smallest program for a universal computer that can reproduce the state in question (for instance, by plotting it with the assumed accuracy). In contrast to the traditional definitions of entropy, algorithmic randomness can be used to measure disorder without any recourse to probabilities. Algorithmic randomness is typically very difficult to calculate exactly but relatively easy to estimate. In large systems, probabilistic en- semble definitions of entropy (e.g., coarse-grained entropy of Gibbs and Boltzmann's entropy $H =\ln 8$; as well as Shannon's information-theoretic entropy) provide accurate estimates of the al- algorithmic entropy of an individual system or its average value for an ensemble. One is thus able to rederive much of thermodynamics and statistical mechanics in a setting very different from the usual. Physical entropy, I suggest, is a sum of (i) the missing information measured by Shannon's formula and (ii) of the algorithmic information content —algorithmic randomness —present in the available data about the system. This definition of entropy is essential in describing the operation of thermodynamic engines from the viewpoint of information gathering and using systems. These Maxwell demon-type entities are capable of acquiring and processing information and therefore can "decide" on the basis of the results of their measurements and computations the best strategy for extracting energy from their surroundings. From their internal point of view the outcome of each measurement is definite. The limits on the thermodynamic efficiency arise not from the ensemble considerations, but rather reflect basic laws of computation. Thus inclusion of algorithmic randomness in the definition of physical entropy allows one to formulate thermodynamics from the Maxwell demon's point of view.
\end{quote}

\normalsize \noindent 11. 
{\"{A}}min Baumeler and Stefan Wolf.
\newblock {Causality--Complexity--Consistency: Can space-time be based on logic
  and computation?}
\newblock In {\em Time in Physics}, pages 69--101. Springer, 2017.
\begin{quote}
\scriptsize \nit{Abstract}:
The difficulty of explaining non-local correlations in a fixed causal structure sheds new light on the old debate on whether space and time are to be seen as fundamental. Refraining from assuming space-time as given a priori has a number of consequences. First, the usual definitions of randomness depend on a causal structure and turn meaningless. So motivated, we propose an intrinsic, physically motivated measure for the randomness of a string of bits: its length minus its normalized work value, a quantity we closely relate to its Kolmogorov complexity (the length of the shortest program making a universal Turing machine output this string). We test this alternative concept of randomness for the example of non-local correlations, and we end up with a reasoning that leads to similar conclusions as in, but is conceptually more direct than, the probabilistic view since only the outcomes of measurements that can actually all be carried out together are put into relation to each other. In the same context-free spirit, we connect the logical reversibility of an evolution to the second law of thermodynamics and the arrow of time. Refining this, we end up with a speculation on the emergence of a space-time structure on bit strings in terms of data-compressibility relations. Finally, we show that logical consistency, by which we replace the abandoned causality, it strictly weaker a constraint than the latter in the multi-party case.
\end{quote}

\subsubsection*{Quantum Algorithmic Complexity}

\normalsize \noindent 12. 
Paul~MB Vit{\'a}nyi.
\newblock Quantum {K}olmogorov complexity based on classical descriptions.
\newblock {\em IEEE Transactions on Information Theory}, 47(6):2464--2479,
  2001.
\begin{quote}
\scriptsize \nit{Abstract}:
We develop a theory of the algorithmic information in bits contained in an individual pure quantum state. This extends classical Kolmogorov complexity to the quantum domain retaining classical descriptions. Quantum Kolmogorov complexity coincides with the classical Kolmogorov complexity on the classical domain. Quantum Kolmogorov complexity is upper-bounded and can be effectively approximated from above under certain conditions. With high probability, a quantum object is incompressible. Upper and lower bounds of the quantum complexity of multiple copies of in- dividual pure quantum states are derived and may shed some light on the no-cloning properties of quantum states. In the quantum situation, complexity is not subadditive. We discuss some relations with ``no-cloning'' and ``approximate cloning'' properties.
\end{quote}

\normalsize \noindent 13. 
Fabio Benatti, Tyll Kr{\"u}ger, Markus M{\"u}ller, Rainer Siegmund-Schultze,
  and Arleta Szko{\l}a.
\newblock Entropy and quantum kolmogorov complexity: A quantum brudno's
  theorem.
\newblock {\em Communications in mathematical physics}, 265:437--461, 2006.
\begin{quote}
\scriptsize \nit{Abstract}:
In classical information theory, entropy rate and algorithmic complexity per symbol are related by a theorem of Brudno. In this paper, we prove a quantum version of this theorem, connecting the von Neumann entropy rate and two notions of quantum Kolmogorov complexity, both based on the shortest qubit descriptions of qubit strings that, run by a universal quantum Turing machine, reproduce them as outputs.
\end{quote}
\subsubsection*{Information Measures}

\normalsize
\noindent 14. 
Charles~H Bennett, P{\'e}ter G{\'a}cs, Ming Li, Paul~MB Vit{\'a}nyi, and
  Wojciech~H Zurek.
\newblock Information distance.
\newblock {\em IEEE Transactions on information theory}, 44(4):1407--1423,
  1998.
\begin{quote}
\scriptsize
\nit{Abstract :} While Kolmogorov complexity is the accepted absolute measure of information content in an individual finite object, a similarly absolute notion is needed for the information distance between two individual objects, for example, two pictures. We give several natural definitions of a universal information metric, based on length of shortest programs for either ordinary computations or reversible (dissipationless) computations. It turns out that these definitions are equivalent up to an additive logarithmic term. We show that the information distance is a universal cognitive similarity distance. We investigate the maximal correlation of the shortest programs involved, the maximal uncorrelation of programs (a generalization of the Slepian–Wolf theorem of classical information theory), and the density properties of the discrete metric spaces induced by the information distances. A related distance measures the amount of nonreversibility of a computation. Using the physical theory of reversible computation, we give an appropriate (universal, antisymmetric, and transitive) measure of the thermodynamic work required to transform one object in another object by the most efficient process. Information distance between individual objects is needed in pattern recognition where one wants to express effective notions of “pattern similarity” or “cognitive similarity” between individual objects and in thermodynamics of computation where one wants to analyze the energy dissipation of a computation from a particular input to a particular output.
\end{quote}

\normalsize \noindent 15. 
Rudi Cilibrasi and Paul~MB Vit{\'a}nyi.
\newblock Clustering by compression.
\newblock {\em IEEE Transactions on Information theory}, 51(4):1523--1545,
  2005.
\begin{quote}
\scriptsize \nit{Abstract} :
We present a new method for clustering based on compression. The method does not use subject-specific features or background knowledge, and works as follows: First, we determine a parameter-free, universal, similarity distance, the normalized compression distance or NCD, computed from the lengths of compressed data files (singly and in pairwise concatenation). Second, we apply a hierarchical clustering method. The NCD is not restricted to a specific application area, and works across application area boundaries. A theoretical precursor, the normalized information distance, co-developed by one of the authors, is provably optimal. How- ever, the optimality comes at the price of using the noncomputable notion of Kolmogorov complexity. We propose axioms to capture the real-world setting, and show that the NCD approximates optimality. To extract a hierarchy of clusters from the distance matrix, we determine a dendrogram (ternary tree) by a new quartet method and a fast heuristic to implement it. The method is implemented and available as public software, and is robust under choice of different compressors. To substantiate our claims of universality and robustness, we report evidence of successful application in areas as diverse as genomics, virology, languages, literature, music, handwritten digits, astronomy, and combinations of objects from completely different domains, using statistical, dictionary, and block sorting compressors. In genomics, we presented new evidence for major questions in Mammalian evolution, based on whole-mitochondrial genomic analysis: the Eutherian orders and the Marsupionta hypothesis against the Theria hypothesis.
\end{quote}

\subsubsection*{Applications}

\normalsize
\nit{Anomaly Detection}
\smallskip

\normalsize \noindent 16. 
Eamonn Keogh, Stefano Lonardi, Chotirat~Ann Ratanamahatana, Li~Wei, Sang-Hee
  Lee, and John Handley.
\newblock Compression-based data mining of sequential data.
\newblock {\em Data Mining and Knowledge Discovery}, 14:99--129, 2007.
\begin{quote}
\scriptsize \nit{Abstract:} The vast majority of data mining algorithms require the setting of many input parameters. The dangers of working with parameter-laden algo- rithms are twofold. First, incorrect settings may cause an algorithm to fail in finding the true patterns. Second, a perhaps more insidious problem is that the algorithm may report spurious patterns that do not really exist, or greatly overestimate the significance of the reported patterns. This is especially likely when the user fails to understand the role of parameters in the data mining process. Data mining algorithms should have as few parameters as possible. A parameter-light algorithm would limit our ability to impose our prejudices, expectations, and presumptions on the problem at hand, and would let the data itself speak to us. In this work, we show that recent results in bioinformatics, learning, and computational theory hold great promise for a parameter-light data-mining paradigm. The results are strongly connected to Kolmogorov complexity theory. However, as a practical matter, they can be implemented using any off-the-shelf compression algorithm with the addition of just a dozen lines of code. We will show that this approach is competitive or superior to many of the state-of-the-art approaches in anomaly/interestingness detection, classification, and clustering with empirical tests on time series/DNA/text/XML/video data- sets. As a further evidence of the advantages of our method, we will demonstrate its effectiveness to solve a real world classification problem in recommending printing services and products.
\end{quote}

\normalsize
\nit{Biology}
\smallskip

 \noindent 17. 
Rudi~L Cilibrasi and Paul~MB Vit{\'a}nyi.
\newblock Fast whole-genome phylogeny of the covid-19 virus sars-cov-2 by
  compression.
\newblock {\em bioRxiv}, pages 2020--07, 2020.
\begin{quote}
\scriptsize \nit{Abstract}: We analyze the phylogeny and taxonomy of the SARS-CoV-2 virus using compression. This is a new alignment-free method called the “normalized compression distance” (NCD) method. It discovers all effective similarities based on Kolmogorov complexity. The latter being incomputable we approximate it by a good compressor such as the modern zpaq. The results comprise that the SARS-CoV-2 virus is closest to the RaTG13 virus and similar to two bat SARS-like coronaviruses bat-SL-CoVZXC21 and bat-SL-CoVZC4. The similarity is quantified and compared with the same quantified similarities among the mtDNA of certain species. We treat the question whether Pangolins are involved in the SARS-CoV-2 virus.
\end{quote}

\smallskip
\nit{Neurosciences}

\normalsize \noindent 18. 
J~Szczepa{\'n}ski, Jos{\'e}~M Amig{\'o}, E~Wajnryb, and MV~Sanchez-Vives.
\newblock Application of lempel--ziv complexity to the analysis of neural
  discharges.
\newblock {\em Network: Computation in Neural Systems}, 14(2):335, 2003.
\begin{quote}
\scriptsize \nit{Abstract}:
Pattern matching is a simple method for studying the properties of information sources based on individual sequences (Wyner et al 1998 IEEE Trans. Inf. Theory 44 2045–56). In particular, the normalized Lempel– Ziv complexity (Lempel and Ziv 1976 IEEE Trans. Inf. Theory 22 75–88), which measures the rate of generation of new patterns along a sequence, is closely related to such important source properties as entropy and information compression ratio. We make use of this concept to characterize the responses of neurons of the primary visual cortex to different kinds of stimulus, including visual stimulation (sinusoidal drifting gratings) and intracellular current injections (sinusoidal and random currents), under two conditions (in vivo and in vitro preparations). Specifically, we digitize the neuronal discharges with several encoding techniques and employ the complexity curves of the resulting discrete signals as fingerprints of the stimuli ensembles. Our results show, for example, that if the neural discharges are encoded with a particular one-parameter method (‘interspike time coding’), the normalized complexity remains constant within some classes of stimuli for a wide range of the parameter. Such constant values of the normalized complexity allow then the differentiation of the stimuli classes. With other encodings (e.g. ‘bin coding’), the whole complexity curve is needed to achieve this goal. In any case, it turns out that the normalized complexity of the neural discharges in vivo are higher (and hence carry more information in the sense of Shannon) than in vitro for the same kind of stimulus.
\end{quote}

\nit{Linguistics}
\smallskip

\normalsize \noindent 19. 
Henry Brighton and Simon Kirby.
\newblock The survival of the smallest: Stability conditions for the cultural
  evolution of compositional language.
\newblock In {\em Advances in Artificial Life: 6th European Conference, ECAL
  2001 Prague, Czech Republic, September 10--14, 2001 Proceedings 6}, pages
  592--601. Springer, 2001.
\begin{quote}
\scriptsize \nit{Abstract}:
Recent work in the field of computational evolutionary linguistics suggests that the dynamics arising from the cultural evolution of language can explain the emergence of syntactic structure. We build on this work by introducing a model of language acquisition based on the Minimum Description Length Principle. Our experiments show that compositional syntax is most likely to occur under two conditions specific to hominids: (i) A complex meaning space structure, and (ii) the poverty of the stimulus.
\end{quote}

\nit{Sociology}
\smallskip

\normalsize \noindent 20. 
Carter~T Butts.
\newblock The complexity of social networks: theoretical and empirical
  findings.
\newblock {\em Social Networks}, 23(1):31--72, 2001.
\begin{quote}
\scriptsize \nit{Abstract}:
A great deal of work in recent years has been devoted to the topic of “complexity”, its measurement, and its implications. Here, the notion of algorithmic complexity is applied to the analysis of social networks. Structural features of theoretical importance — such as structural equivalence classes — are shown to be strongly related to the algorithmic complexity of graphs, and these results are explored using analytical and simulation methods. Analysis of the complexity of a variety of empirically derived networks suggests that many social networks are nearly as complex as their source entropy, and thus that their structure is roughly in line with the conditional uniform graph distribution hypothesis. Implications of these findings for network theory and methodology are also discussed.
\end{quote}

\nit{Ecology}
\smallskip

\normalsize \noindent 21. 
Vasilis Dakos and Fernando Soler-Toscano.
\newblock Measuring complexity to infer changes in the dynamics of ecological
  systems under stress.
\newblock {\em Ecological Complexity}, 32:144--155, 2017.
\begin{quote}
\scriptsize \nit{Abstract}:
Despite advances in our mechanistic understanding of ecological processes, the inherent complexity of real-world ecosystems still limits our ability in predicting ecological dynamics especially in the face of ongoing environmental stress. Developing a model is frequently challenged by structure uncertainty, unknown parameters, and limited data for exploring out-of-sample predictions. One way to address this challenge is to look for patterns in the data themselves in order to infer the underlying processes of an ecological system rather than to build system-specific models. For example, it has been recently suggested that statistical changes in ecological dynamics can be used to infer changes in the stability of ecosystems as they approach tipping points. For computer scientists such inference is similar to the notion of a Turing machine: a computational device that could execute a program (the process) to produce the observed data (the pattern). Here, we make use of such basic computational ideas introduced by Alan Turing to recognize changing patterns in ecological dynamics in ecosystems under stress. To do this, we use the concept of Kolmogorov algorithmic complexity that is a measure of randomness. In particular, we estimate an approximation to Kolmogorov complexity based on the Block Decomposition Method (BDM). We apply BDM to identify changes in complexity in simulated time-series and spatial datasets from ecosystems that experience different types of ecological transitions. We find that in all cases, KBDM complexity decreased before all ecological transitions both in time-series and spatial datasets. These trends indicate that loss of stability in the ecological models we explored is characterized by loss of complexity and the emergence of a regular and computable underlying structure. Our results suggest that Kolmogorov complexity may serve as tool for revealing changes in the dynamics of ecosystems close to ecological transitions.
\end{quote}

\nit{Ethology}
\smallskip

\normalsize \noindent 22. 
Boris Ryabko and Zhanna Reznikova.
\newblock The use of ideas of information theory for studying ``language'' and
  intelligence in ants.
\newblock {\em Entropy}, 11(4):836--853, 2009.
\begin{quote}
\scriptsize \nit{Abstract}:
In this review we integrate results of long term experimental study on ant ``language'' and intelligence which were fully based on fundamental ideas of Information Theory, such as the Shannon entropy, the Kolmogorov complexity, and the Shannon’s equation connecting the length of a message (l) and its frequency (p), i.e., $l = - \log p$ for rational communication systems. This approach enabled us to obtain the following important results on ants’ communication and intelligence: (i) to reveal “distant homing” in ants, that is, their ability to transfer information about remote events; (ii) to estimate the rate of information transmission; (iii) to reveal that ants are able to grasp regularities and to use them for “compression” of information; (iv) to reveal that ants are able to transfer to each other the information about the number of objects; (v) to discover that ants can add and subtract small numbers. The obtained results show that information theory is not only excellent mathematical theory, but many of its results may be considered as Nature laws.
\end{quote}

\newpage
\bibliographystyle{unsrt}
\bibliography{/Applications/TeX/ref}

\end{document}